\documentclass[11pt, draftclsnofoot]{IEEEtran}
\usepackage{setspace,multirow, verbatim, amsfonts, graphicx, amsmath, amsbsy, amssymb, epsfig, url}
\usepackage{booktabs}
\onecolumn

\newcommand{\Sb}{{\mathbf S}}

\newcommand{\ab}{{\mathbf a}}

\newcommand{\eb}{{\mathbf e}}

\newcommand{\vb}{{\mathbf v}}
\newcommand{\wb}{{\mathbf w}}
\newcommand{\xb}{{\mathbf x}}
\newcommand{\yb}{{\mathbf y}}
\newcommand{\zb}{{\mathbf z}}

\newcommand{\gammab}{{\boldsymbol{\gamma}}}

\renewcommand{\Re}{{\mathbb{R}}}
\newcommand{\Rd}{{\mathbb R}}

\newcommand{\zerob}{\boldsymbol{0}}

\newtheorem{theorem}{Theorem}[section]

\DeclareMathOperator*{\argmin}{arg\,min}

\title{Improving M-SBL for Joint Sparse Recovery using a Subspace Penalty}

\author{ Jong Chul Ye,  Jong Min Kim and Yoram Bresler}

\begin{document}
%
\maketitle
\vspace{-1.5cm}
\begin{abstract}
\baselineskip 0.2in
The  multiple measurement vector  problem (MMV) is a generalization of the compressed sensing problem that addresses the recovery of a set of  jointly sparse signal vectors.
One of the important contributions of this paper is to reveal that the seemingly least related  state-of-art MMV joint sparse recovery algorithms - M-SBL (multiple sparse Bayesian learning) and subspace-based hybrid greedy algorithms - have a very important link. More specifically,  we show that  replacing the $\log\det(\cdot)$  term in M-SBL  by a  rank proxy that exploits the spark reduction property  discovered in subspace-based joint sparse recovery algorithms, provides significant improvements. In particular, if we use the  Schatten-$p$ quasi-norm as the corresponding rank proxy, the global minimiser of the proposed algorithm
becomes identical to the true solution as $p \rightarrow 0$.
Furthermore, under the same regularity conditions,  we show that the convergence to a local minimiser is guaranteed
using an alternating   minimization algorithm that has closed form expressions for each of the minimization steps, which are convex. Numerical simulations under a variety of scenarios   in terms of SNR,  and condition number of the signal amplitude matrix demonstrate that the proposed algorithm  consistently outperforms M-SBL and other state-of-the art algorithms.
\end{abstract}

\begin{keywords}
Compressed sensing, joint sparse recovery, multiple measurement vector problem,  subspace method,  M-SBL, generalized MUSIC criterion, rank proxy, Schatten-$p$ norm
\end{keywords}
%


\noindent{
\vspace{-0.2cm}Correspondence to:\\
\vspace{-0.2cm}Jong Chul Ye,  Ph.D. ~~Professor \\\
\vspace{-0.2cm}Dept. of Bio and Brain Engineering,  KAIST \\
\vspace{-0.2cm}291 Daehak-ro Yuseong-gu, Daejon 305-701, Republic of Korea \\
\vspace{-0.2cm}Email: jong.ye@kaist.ac.kr \\
\vspace{-0.2cm}Tel: 82-42-350-4320 \\
\vspace{-0.2cm}Fax: 82-42-350-4310 \\
 }

\IEEEpeerreviewmaketitle

\newpage

\section{Introduction}
\label{sec:intro}

The  multiple measurement vector  problem (MMV) is a generalization of the compressed sensing problem, which addresses the recovery of a set of sparse signal vectors
that share a common support \cite{Kim2010CMUSIC, chen2006trs,cotter2005ssl,Mishali08rembo,Berg09jrmm,Lee2010SAMUSIC}.  
In the MMV model, let
$m$ and $N$ denote the number of sensor elements and snapshots,
respectively; and $n>m$ denote the length of the signal vectors. Then, for a given noisy observation matrix
$Y=[\mathbf{y}_1,\cdots,\mathbf{y}_N]\in\mathbb{C}^{m\times N}$ and a sensing matrix $A\in
\mathbb{C}^{m\times n}$, the multiple measurement vector (MMV)
problem can be formulated as:
\begin{eqnarray}\label{eqdefmmv}
{\rm minimize}~~~\|X\|_0\\
{\rm subject~to}~~~\|Y-AX\|_F < \delta, \notag
\end{eqnarray}
where  $\xb_j \in \Rd^n$ is the $j$-th signal, $X=[\mathbf{x}_1,\cdots,\mathbf{x}_N]\in\mathbb{R}^{n\times
N}$, $\mathbf{x}^i$ is the $i$-th row of $X$, and $\|X\|_0=|{\rm supp}X|$, where ${\rm supp}X=\{1\leq i\leq n
: \mathbf{x}^i\neq 0\}$ is the set of indices of nonzero rows in $X$. The Frobenius norm is used to measure the discrepancy between the data and the model.
Classically, pursuit
algorithms such as alternating minimization algorithm (AM) and  MUSIC (multiple signal classification) algorithm \cite{Feng97},
 S-OMP (simultaneous orthogonal matching pursuit)
\cite{tropp2006ass,chen2006trs},  M-FOCUSS
\cite{cotter2005ssl}, randomized algorithms such as REduce MMV
and BOost (ReMBo)\cite{Mishali08rembo}, and model-based compressive
sensing using block-sparsity
\cite{eldar2009compressed,baraniuk2010model} have  been applied
to the MMV problem. 

An algebraic bound for the recoverable sparisity level has been theoretically studied by  Feng and Bresler \cite{Feng97}, and  by Chen and
Huo \cite{chen2006trs} for noiseless measurement $Y$. More specifically,  if
 $X\in\mathbb{R}^{n\times N}$ satisfies
$AX=Y$ and
\begin{equation}\label{l0-bound-mmv}
\|X\|_0<  \frac{{\rm spark}(A)+{\rm rank}(Y)-1}{2} ,
\end{equation}
where  ${\rm spark}(A)$ denotes the smallest number
of linearly dependent columns of $A$,  then $X$ is the unique solution of (\ref{eqdefmmv}). This indicates that the recoverable sparsity level may increase with an
increasing number of measurement vectors. 
Indeed,  for noiseless measurement,  a MUSIC  algorithm by Feng and Bresler \cite{Feng97} is shown to achieve the performance limit when the measurement matrix is full rank.
However,  except for  MUSIC in full rank cases, the performance of the aforementioned classical MMV  algorithms   is not generally satisfactory,  falling far short of \eqref{l0-bound-mmv}  even for
the noiseless case, when only a finite number of snapshots are available.

In a noisy environment,  Obozinski {\em et al} showed that
 a near
optimal  sampling rate reduction   up to ${\rm rank}(Y)$  can be achieved using  $l_1/l_2$ mixed norm penalty \cite{obozinski2011support}. 
A similar gain was observed in computationally inexpensive greedy approaches such as compressive MUSIC (CS-MUSIC) \cite{Kim2010CMUSIC} and subspace augmented MUSIC (SA-MUSIC) \cite{Lee2010SAMUSIC}.  More specifically, Kim {\em { et al}} \cite{Kim2010CMUSIC} and Lee {\em { et al}} \cite{Lee2010SAMUSIC} independently showed that a class of hybrid greedy algorithms  that combine greedy steps with a so called {\em generalized MUSIC}  subspace criterion \cite{Kim2010CMUSIC}, or equivalently, with subspace augmentation \cite{Lee2010SAMUSIC},  can reduce the required number of measurements by up to ${\rm rank}(Y)$  in noisy environment. Furthermore, using a large system MMV model,  Kim {\em { et al}} further showed that for an i.i.d. Gaussian sensing matrix, their algorithm can asymptotically achieve the algebraic performance limit when ${\rm rank}(Y)$ increases with a particular  scaling law \cite{Kim2010CMUSIC}.
Lee {\em et al} \cite{Lee2010SAMUSIC} also showed that MUSIC can do this in the noisy case and full rank, non-asymptotically with finite data, and for realistic Fourier sensing matrices.

While  the aforementioned mixed norm  approach and subspace based greedy  approaches provide theoretical performance guarantees, there also exist a very different class of powerful MMV algorithms that are based on empirical Bayesian and Automatic Relevance Determination (ARD) principles from machine learning. Among these, the so-called multiple sparse Bayesian learning (M-SBL)  algorithm is best known \cite{wipf2007ebs}. 
Even though M-SBL is  more computationally expensive than  greedy algorithms such as CS-MUSIC or SA-MUSIC,  empirical results show that M-SBL is quite robust to noise  and to unfavorable restricted isometry property constant (RIC) of the sensing matrix \cite{KimSeqCS2012}. Moreover, M-SBL is more competitive than mixed norm approaches.
Since  Bayesian approaches are very different from classical compressed sensing, such high performance appears mysterious at first glance.  However, a recent  breakthrough by Wipf {\em { et al}}  unveiled that M-SBL can be converted to a standard compressed sensing framework with  an additional $\log\det(\cdot)$ (log determinant) penalty -   a non-separable sparsity inducing prior  \cite{wipf2011latent}.  The presence of the non-separable penalty term  is  so powerful that M-SBL performs almost as well as MUSIC.  
However,
the guarantee only applies to the full row rank case  of $X$ with  noise-free measurement vectors \cite{wipf2006bayesian,wipf2007ebs}.  However, despite its excellent  performance, compared to the mixed norm approaches or subspace greedy algorithms,  other than the  work by Wipf {\em { et al}}  \cite{wipf2011latent},  the fundamental theoretical analysis of M-SBL has been limited.

Therefore, one of  the main goals of this paper is to continue the effort by Wipf {\em { et al}}  \cite{wipf2011latent} and analyze the origin of the high performance of M-SBL, as well as  to investigate its limitations.
One of the important contributions of this paper is to show that the seemingly least related algorithms - M-SBL and subspace-based hybrid greedy algorithms - have a very important link. More specifically,  we show that  the $\log \det(\cdot)$  term  in M-SBL is  a proxy for the rank of a partial sensing matrix corresponding to  the true support. 
We then show that minimising the rank  that was discovered in subspace-based hybrid greedy algorithm to exploit the spark reduction property of MMV can indeed provide a true solution for  the MMV problem.
Accordingly,  replacing $\log \det(\cdot)$  term  in M-SBL   by 
a  Schatten-$p$ quasi-norm rank proxy provides significant performance improvements.

%

 The resulting new algorithm is no longer Bayesian  due to the use of a  {\em deterministic} penalty based on a geometric argument, so we call the new algorithm {\em subspace-penalized sparse learning (SPL)} by excluding term ``Bayesian''.  We show that as $p\rightarrow 0$ in the Schatten $p$-norm rank proxy, the global minimizers 
 of the SPL cost function are  identical to those of the original $l_0$ minimization problem.
   Furthermore,  we show that SPL can be easily implemented as an alternating minimization approach.

Using numerical simulations, we demonstrate that compared to the current state-of-art MMV algorithms such as mixed norm approaches, M-SBL, CS-MUSIC/SA-MUSIC, and sequential CS-MUSIC \cite{KimSeqCS2012}, SPL provides superior  recovery performance. Moreover, the results show that  SPL is very robust to noise, and  to the condition number of the unknown signal matrix.  

\subsection{Notation}

Throughout the paper, $\xb^i$ and $\xb_j$ correspond to the $i$-th
row and the $j$-th column of matrix $X$, respectively.  The $(i,j)$ element of $X$ is represented by $x_{ij}$.
When $S$ is an
index set, $X^S$, and $A_S$ correspond to a submatrix collecting
corresponding rows of $X$ and columns of $A$, respectively. 
For a matrix $A$,   ${\rm Tr}(A)$  is  the trace of a matrix $A$,   $A^*$ is its adjoint, $A^\dag$ denotes the Penrose-Moor psuedo-inverse,
  $|A|$ refers the determinant,   $R(A)$ denotes the range space of $A$, and $P_A$ (or $P_{R(A)}$) and $P_A^\perp$  (or $P_{R(A)}^\perp$) are the projection on the range space and its orthogonal complement, respectively.
The vector $\eb_i$ denotes an elementary unit vector whose $i$-th element is 1, and $I$ denotes an identity matrix.

A sensing matrix $A\in\mathbb{R}^{m\times n}$ is said to have a
$k$-restricted isometry property (RIP) if there exist left and right
RIP constants $0\leq \delta^L_k, \delta^R_k<1$  such that
$$(1-\delta^L_k)\|\mathbf{x}\|^2\leq \|A\mathbf{x}\|^2\leq (1+\delta^R_k)\|\mathbf{x}\|^2$$
for all $\mathbf{x}\in\mathbb{R}^n$ such that $\|\mathbf{x}\|_0\leq
k$. A single RIP constant $\delta_k = \max\{\delta^L_k,
\delta^R_k\}$ is often referred to as the RIP constant. 

\section{M-SBL: A Review}

\subsection{Algorithm Description}
   Under appropriate  assumptions of noise and signal Gaussian statistics,  one can show that M-SBL  minimizes the following cost function in a so-called $\gammab$ space \cite{wipf2007ebs}:
\begin{eqnarray}\label{eq:costgamma}
{\cal L}^\gamma(\gammab) = {\rm Tr}\left(\Sigma_y^{-1}Y Y^*\right)  +  N \log|\Sigma_y| 
\end{eqnarray}
where 
\begin{eqnarray}
\Sigma_y = \lambda I+A\Gamma A^* &, & \Gamma = \mathrm{diag}(\gammab) \ .
\end{eqnarray}
With an estimate of  $\Gamma$, which typically has a nearly sparse diagonal and may be thresholded to  be exactly sparse,  the solution of M-SBL  is given by
\begin{eqnarray}\label{eq:solution}
X = \Gamma A^*(\lambda I+A \Gamma A^* )^{-1} Y \  .
\end{eqnarray}
One of the most important contributions   by Wipf is that the minimization problem of the cost function \eqref{eq:costgamma} can be equivalently represented as the following standard sparse recovery framework \cite{wipf2011latent}:
\begin{eqnarray}\label{eq:minX}
\min_X {\cal L}^\xb(X) ,\quad {\cal L}^\xb(X)=\|Y-AX\|_F^2 + \lambda g_{msbl}(X)
\end{eqnarray}
where  $g_{msbl}(X) $ is  a   penalty given by
\begin{eqnarray}\label{eq:gII}
g_{msbl}(X) = \min_{\gammab\geq \zerob}  G(X;\gammab)
\end{eqnarray}
where
\begin{eqnarray}
G(X;\gammab) \equiv
{\rm Tr}\left(X^*\Gamma^{-1}X\right) 
+   N \log|\lambda I+A\Gamma A^*|  \ . 
\end{eqnarray}
Wipf {\em {et al}} \cite{wipf2011latent} gave a heuristic argument showing that $g_{msbl}$ corresponds to a non-separable sparsity promoting penalty, and
proposed the following alternating minimization approach to solve the minimization problem \eqref{eq:minX}.

\subsubsection{Step 1: Minimization with respect to $X$}

For a given estimate  $\gammab^{(t)}$ at the $t$-th iteration,  we can find a closed form solution for $X$ in \eqref{eq:minX}:
$$X^{(t)}=\Gamma^{(t)} A^*(\lambda I+A\Gamma^{(t)}A^*)^{-1}Y,\quad \Gamma^{(t)}=\textrm{diag}(\gammab^{(t)}).$$
For large scale problems, this can be computed using a standard conjugate gradient algorithm with an appropriate preconditioner.

\subsubsection{Step 2: Minimization with respect to $\gammab$}

In this step, for a given $X^{(t)}$ we need to solve the following minimization problem:
$$\gammab^{(t+1)}=\argmin_{\gammab\geq \zerob} G(X^{(t)},\gammab)$$
where 
\begin{equation}\label{eq:G}
G(X^{(t)},\gammab)=  {\rm Tr}\left(X^{(t)*}\Gamma^{-1}X^{(t)}\right)  +   N \log|\Sigma_y| \ . 
\end{equation}

Wipf {\em et al} \cite{wipf2011latent} 
 find the solution to $\nabla G(X^{(t)},\gammab)=\zerob$.
More specifically,   the derivative with respect to each component is given by
\begin{eqnarray}\label{eq:coord}
\frac{\partial G(X^{(t)},\gammab) }{\partial \gamma_i} = -\frac{\sum_j |x^{(t)}_{ij}|^2}{\gamma_i^2}+N \ab_i^H(\lambda I+A\Gamma A^*)^{-1}\ab_i
\end{eqnarray}
since $$\frac{\partial |\Sigma_y|}{\partial \gamma_i}=|\Sigma_y|~ \mathrm{Tr}\left(\Sigma_y^{-1}\frac{\partial \Sigma_y}{\partial \gamma_i}\right) \  .$$
Setting  the derivative to zero after fixing $\Gamma:=\Gamma^{(t)}$, this observation leads to the following fixed point update of $\gammab$:
\begin{eqnarray}\label{eq:msbl_gamma}
\gamma_i^{(t+1)}= \left(\frac{\frac{1}{N}\sum_j |x^{(t)}_{ij}|^2}{\ab_i^H(\lambda I+A\Gamma^{(t)}A^*)^{-1}\ab_i}\right)^{\frac{1}{2}} \  .
\end{eqnarray}

\subsection{Role of the Non-separable Penalty in M-SBL}

In order to develop a new joint sparse recovery algorithm that improves on M-SBL,  we provide here a new interpretation of  the role of the regularization term in M-SBL.
Note that 
 due to the non-negativity constraint for $\gammab$,  a critical solution to the minimization problem in \eqref{eq:gII} should satisfy the following first order Karush-Kuhn-Tucker (KKT) necessary conditions \cite{ChZa96}:
\begin{eqnarray*} 
\gamma_i \frac{\partial G(X ,\gammab)}{\partial \gamma_i}  &=&   \gamma_i \left(-\frac{ \sum_j |x_{ij} |^2}{\gamma_i^2} + N\ab_i^H(\lambda I+A\Gamma A^*)^{-1}\ab_i
\right)= 0,\quad \forall ~i  
\end{eqnarray*}
Hence, as $\lambda \rightarrow 0$, this leads to the following fixed point equation:
\begin{equation*}
\lim_{\lambda \rightarrow 0}\gamma_i = \lim_{\lambda \rightarrow 0} \frac{\frac{1}{N} \|\xb^i\|^2}{\gamma_i \ab_i^*(\lambda I+A\Gamma A^*)^{-1}\ab_i} \ .
\end{equation*}
If $\|\gammab\|_0<m$, using  the matrix inversion lemma,  we have
\begin{eqnarray}\label{eq:matinv}
\lambda (\lambda I+A\Gamma A^*)^{-1} &=& I-   A\Gamma^{\frac{1}{2}} \left(\lambda I  + \Gamma^{\frac{1}{2}} A^*A\Gamma^{\frac{1}{2}}\right)^{-1}\Gamma^{\frac{1}{2}} A^* \nonumber\\
&=& P^{\perp}_{A_S} + \lambda ( A_S^{\dag})^* \left(\lambda (A_S^*A_S)^{-1}+\Gamma_S\right)^{-1}A_S^\dag \ ,
\end{eqnarray}
where $S={\rm supp}(\gammab)$ denotes a nonzero support of $\gammab$ and $P_{A_S}$ denotes the orthogonal projection on the span of the columns of
$A$ indexed by $S$.
Accordingly,
$$\lim_{\lambda \rightarrow 0} { \ab_i^*(\lambda I+A\Gamma A^*)^{-1}\ab_i} = \ab_i^*(A_S^\dag)^*\Gamma_S^{-1} A_S^\dag \ab_i = \frac{1}{\gamma_i},  \quad i \in S.$$
Therefore, we have
\begin{equation}\label{eq:gamma_proj}
\lim_{\lambda \rightarrow 0}\gamma_i = \frac{1}{N} \|\xb^i\|^2, \quad i \in S.
\end{equation}
Substituting \eqref{eq:gamma_proj}  into \eqref{eq:gII} yields
\begin{eqnarray}
g_{msbl}(X) &= & \min_{\gammab\geq \zerob}  {\rm Tr}\left(X^*\Gamma^{-1}X\right) 
+   N \log|\lambda I+A\Gamma A^*|  \notag\\
&\approx& N~|S| + N \log|\lambda I + A \Gamma A^*| 
\notag\\
&=& N \|\gammab\|_0 + N \log|\lambda I + A \Gamma A^*|  \label{eq:gII2}
\end{eqnarray}
Note that the first term in \eqref{eq:gII2} imposes  row sparsity on $X$ since $\gamma_i=0$ for $\|\xb^i\|=0$ due to \eqref{eq:gamma_proj}.  Hence, the first term of M-SBL penalty is in fact  $\|X\|_0$. Then, what is the meaning of the $\log|\cdot|$ term ? 
 Wipf {\em et al} \cite{wipf2011latent} showed  that the superior performance of the M-SBL is owing to the non-separability of the term
 $\log|\lambda I + A\Gamma A^*|$ with respect to $\gammab$, which 
  can avoid many local minimizers.
In addition to this interpretation, the following section shows another  important geometric implications of the $\log\det(\cdot)$ term.

%

\section{Subspace-Penalized Sparse Learning}


\subsection{Key Observation}

In this section, we provide another interpretation of the M-SBL penalty, which  suggests  a new algorithm called subspace-penalized sparse learning (SPL) that overcomes the limitation of M-SBL.  
Note that  for any matrix $Z\in \Rd^{m\times n}$ with $m\leq n$, we have 
\begin{eqnarray}\label{eq:svd_proxy}
\log |ZZ^*+\lambda I| &=& \sum_{i=1}^m  \log (\sigma^2_i(Z)+\lambda),
\end{eqnarray}
where $\sigma_i(Z)$ denote the singular values of $Z$.  Therefore, the $\log\det()$ function is a concave  proxy for nonzero singular values, hence in the limit of $\lambda \rightarrow 0$,
\eqref{eq:svd_proxy} acts as a
proxy for ${\rm rank}(Z)$  \cite{mohan2010iterative,fazel2003log}.
This leads us  to  another interpretation:  the penalty term in M-SBL is equivalent to
\begin{equation}\label{eq:sparse_rank}
\lim_{\lambda \rightarrow 0} g_{msbl}(X) =  N  \|\gammab\|_0 +  N ~{\rm Rprox}(A\Gamma^{\frac{1}{2}})
\end{equation}
where ${\rm Rprox}(\cdot)$ dentoes a rank proxy. Thus,   the penalty 
 simultaneously imposes  the row sparsity of $X$  as well as the low rank of the matrix $A\Gamma^{\frac{1}{2}}$.
By inspection,  the first sparsity penalty term in \eqref{eq:sparse_rank} is quite intuitive, but it is not clear why ${\rm Rank}(A\Gamma^{\frac{1}{2}})$ needs to be minimized. 

In fact,  the main contribution of this paper is that we need to replace the second term, ${\rm Rprox}(A\Gamma^{\frac{1}{2}})$, by geometrically more intuitive rank proxy as follows:
\begin{equation}\label{eq:sparse_qrank}
g_{SPL}(X) =  N  \|\gammab\|_0 +  N ~{\rm Rprox}(Q^*A\Gamma^{\frac{1}{2}})
\end{equation}
where $Q$ denotes a basis for the noise subspace denoted as $R(Q) = R^\perp(Y)$.   In the following, we will  describe in detail how we arrive at the new rank penalty.

\subsection{Subspace Criteria}

A  solution $X$ for $Y=AX$ that  satisfies  $\|X\|_0\leq m$ is called a {\em basic feasible solution}  (BFS) \cite{wipf2011latent}.
Among BFSs, a solution of the following $l_0$ MMV problem is called {\em maximally sparse} solution:
\begin{eqnarray}\label{eq:prob}
(P0):\quad \min_X \|X \|_0 &,& \mbox{subject to $Y=AX$}.
\end{eqnarray}
To address \eqref{eq:prob},
subspace-based greedy algorithms such as CS-MUSIC  \cite{Kim2010CMUSIC} and SA-MUSIC \cite{Lee2010SAMUSIC} exploit
the {\em spark reduction principle} or
 or an equivalent subspace criterion using an {\em augmented signal subspace}.  More specifically, if  $r={\rm rank}(Y)$
 and $k$ denotes the number of the non-zero rows,
 the algorithms first estimate $k-r$ partial support index $I_{k-r}$, then the remaining $r$ components of the support are found using the subspace criterion.
%

One of the main contributions of this paper is that this two step approach is not necessary. Instead,   for noiseless measurements, 
  a direct  minimization of the rank of $Q^*A_I$  with respect to index $I$, $|I|\geq k$,    still guarantees to obtain the true support as shown in Theorem~\ref{thm:spark}. We believe that this is an extremely powerful result that provides an important clue to overcome the limitation of existing greedy subspace methods  \cite{Kim2010CMUSIC,Lee2010SAMUSIC}  .

\begin{theorem}\label{thm:spark}
Assume that $A\in \mathbb{R}^{m\times n}$, $X\in\mathbb{R}^{n\times r}$, $Y\in\mathbb{R}^{m\times r}$ satisfy $AX=Y$, where $\|X\|_0=k$, and the columns of $Y$ are linearly independent. 
If $A$ satisfies a RIP condition $0\leq \delta^L_{2k-r+1}(A)<1$, then  we have
$$k-r=\min_{|I|\geq k }{\rm rank}\left(Q^*A_I\right),$$
and $${\rm supp}X=\arg\min_{|I|\geq k }{\rm rank}\left(Q^*A_I\right).$$ 
%
\end{theorem}
\begin{proof}
Since $\min_{|I|=l}{\rm rank}(Q^*A_I)$ is a nondecreasing function of $l$, we may consider $\min_{|I|=k}{\rm rank}(Q^*A_I)$.
By the rank-nullity Theorem, ${\rm dim}(N(Q^{*}A_I))+{\rm dim}(R(Q^{*}A_I))=k$ for $I\subset \{1,\cdots,n\}$ with $|I|=k$. Furthermore, because $N(Q^*) = Y$,
\begin{eqnarray}\label{null-qa}
N(Q^{*}A_I)=\{\vb\in\mathbb{R}^k:A_I\vb\in R(Y)\} 
\end{eqnarray}
and because $r\leq k$, $N(A_I)=\{0\}$ by the RIP condition $0\leq \delta_{2k-r+1}(A)<1$. 
Since $N(A_I)=\{0\}$, we have ${\rm dim}(N(Q^{*}A_I))={\rm dim}(R(Y)\cap R(A_I))$ so that 
\begin{eqnarray}\label{null-qa-dim}
{\rm dim}(N(Q^{*}A_I))={\rm dim}(R(Y)\cap R(A_I))\leq {\dim}(R(Y))=r,
\end{eqnarray}
which also implies that ${\rm rank}(Q^{*}A_I)\geq k-r$ for any $|I|=k$. Hence,  denoting ${\rm supp}(X) = S$, it is enough to show that 
\begin{equation}\label{proxy-s}
{\rm rank}(Q^{*}A_S)=k-r
\end{equation}
and 
\begin{equation}\label{proxy-nots}
{\rm rank}(Q^{*}A_I)>k-r,~{\rm for}~|I|=k~{\rm and}~I\neq S.
\end{equation}
First we will show that \eqref{proxy-s} holds. Because $N(A_S)=\{0\}$ and $Y=A_SX^S$, we have that ${\rm rank}(Y)={\rm rank}(X^S)$ or ${\rm dim}(R(X^S))=r$. Also, since $Y=A_SX^S,$ by \eqref{null-qa}, we have $R(X^S)\subset N(Q^{*}A_S)$.  Hence $\mathrm{dim}(N(Q^*A_S)) \geq r$. On the other hand, by \eqref{null-qa-dim}, the dimension of $N(Q^{*}A_S)$ is at most $r$, which implies 
${\rm dim}(N(Q^{*}A_S))=r$, so that we have ${\rm rank}(Q^{*}A_S)=k-r.$\\
Then, \eqref{proxy-nots} is  the only a remaining part to prove. Suppose that we have an index set $I\subset \{1,\cdots,n\}$ such that $|I|=k$ and ${\rm rank}(Q^{*}A_I)=k-r.$ Then, it
must  hold that 
\begin{equation*}
{\rm dim}(N(Q^{*}A_I))={\rm dim}(R(Y)\cap R(A_I))=r={\rm dim}R(Y).
\end{equation*}
It follows that $R(Y)=R(Y)\cap R(A_I)\subset R(A_I)$. Then, for each column $\yb_i$ of $Y$, $\yb_i\in R(Y)\subset R(A_I)$ so that there exists $\zb_i \in  \mathbb{R}^{|I|}$ 
such that $A_I\zb_i=\yb_i$. Then, $A_I [\zb_1~\cdots~\zb_r]=[\yb_1~\cdots~\yb_r]=Y$ so that there is a $\tilde{X}
\in\mathbb{R}^{n\times r}$ such that $Y=A\tilde{X}$ with ${\rm supp}\tilde{X}\subset I$. 

Since ${\rm rank}(X)=r$ and $(\xb^i)^T= 0$ for any $i\notin S= {\rm supp}X$, it follows that $X^S$ has rank $r$. Hence, because the row rank of a matrix equals its column rank, $X^S$ must have $r$ linearly independent rows. Therefore, there is a subset $Z$ of ${\rm supp}X$ such that $|Z|=r$, and the rows of $X^Z$ are linearly independent.

Since $X^Z\in\mathbb{R}^{r\times r}$, for every $i\in Z$ there is a nonzero vector $\wb_i\in N(X^{Z\setminus\{i\}})$, so that we have $\|X\wb_i\|_0\leq k-r+1$ and $i\in {\rm supp}(X\wb_i)$, since $\xb^i\wb_i\neq 0$ by the linearly independence of the rows of $X^Z$. Since $Y=AX=A\tilde{X}$, we have $A(X\wb_i-\tilde{X}\wb_i)=0$. 

Now, because ${\rm supp}(\tilde{X})\subset I$, we have ${\rm supp}(\tilde{X}\wb_i)\subset {\rm supp}(\tilde{X})\subset I$. Hence $\|\tilde{X}\wb_i\|\leq |I|=k$. It follows that 
$$\|X\wb_i-\tilde{X}\wb_i\|_0\leq k-r+1+|I|=2k-r+1.$$ 
Hence, by the RIP of $A$, we must have $X\wb_i=\tilde{X}\wb_i$. Since ${\rm supp}(\tilde{X}\wb_i)\subset I$, we also have ${\rm supp}(\tilde{X}\wb_i)={\rm supp}(X\wb_i)\subset I$, which implies that $i\in {\rm supp}(X\wb_i)\subset I$. Since $i$ can be any element in ${\rm supp}X$, we have $i\in I$ for any $i\in {\rm supp}X$. It implies that ${\rm supp}X\subset I$ so that $I={\rm supp}X$ since $|I|=|{\rm supp}X|=k.$   Hence, in order to satisfy ${\rm rank}(Q^{*}A_I)=k-r$ and $|I|=k$, we must have $I={\rm supp}X.$
\end{proof}

\subsection{The SPL Penalty }

 Note that minimizing ${\rm rank}\left(Q^*A\Gamma^{\frac{1}{2}} \right)$ with respect to $\gammab$ is equivalent to finding the index set $I$ that minimizes ${\rm rank}\left(Q^*A_I\right)$.  Hence,
Theorem~\ref{thm:spark} implies that minimizing ${\rm rank}(Q^*A\Gamma^{\frac{1}{2}})$ under the constraint $\|\gammab\|_0 \geq k$ will find $\gammab_*$ that has non-zero values for indices  corresponding to ${\rm supp}X_*$, where $\|X_*\|_0=k$ and $Y=AX_*$. This observation leads to  the second term in the SPL penalty  of \eqref{eq:sparse_qrank} as a rank proxy to exploit this  geometric finding. 

Moreover, rather than just using $\log\det(\cdot)$ as in M-SBL,  in this paper, we use more general family of rank proxies that still satisfy our goals. Specifically,   our rank proxy is based on {\em Schatten-$p$ quasi norm} with $0<p\leq 1$ that includes the popular nuclear norm as a special case.   For  a matrix $W \in \Rd^{m\times n}$,  the Scatten $p$-norm proxy for the rank is defined as
\begin{equation}
{\rm Tr}|W|^p= {\rm Tr}\left(\left(WW^* \right)^{\frac{p}{2}}\right)= \sum_{i=1}^{m} \sigma_i^p(W) \ ,
\end{equation}
which  corresponds to the nuclear norm when $p=1$.
Following the derivation that leads to  \eqref{eq:gII},   
we propose the following SPL penalty:
\begin{eqnarray}
g_{SPL}(X) \equiv \min\limits_{\gammab\geq \zerob}  {\cal G}_{SPL}(\gammab,  X)
\end{eqnarray}
where
\begin{equation}\label{eq:Gssl2}
{\cal G}_{SPL}(\gammab,X) =  {\rm Tr}\left(X^{*}\Gamma^{-1}X\right) +  N {\rm Tr} \left( \left(Q^*A\Gamma A^*Q\right)^{\frac{p}{2}}  \right)\  .
\end{equation}
Using the proposed SPL penalty, we  formulate the following noiseless  SPL  minimization problem:
\begin{eqnarray}
 \min_X g_{SPL}(X), & \mbox{subject to $Y=AX$ }\label{eq:noiselessSPL}  \  .
\end{eqnarray}
Note that ${\rm Tr} \left(\left(Q^*A\Gamma A^*Q\right)^{\frac{p}{2}}\right)$ is a concave function with respect to its singular values,  so we can find its  convex conjugate:
\begin{eqnarray}
{\cal G}_{SPL}(\gammab,X) \equiv \min\limits_{\Psi \in \Sb_{0+}}  {\mathfrak G }_{SPL}(\gammab, X, \Psi)
\end{eqnarray}
 where ${\mathfrak G}_{SPL}(\gammab, X, \Psi)$ is given  as
\begin{eqnarray}\label{eq:gSPL}
{\mathfrak G}_{SPL}(\gammab, X,\Psi) \equiv  {\rm Tr}\left(X^*\Gamma^{-1}X\right) 
+   Np \left( {\rm Tr}\left( (Q^*A\Gamma A^*Q)\Psi \right) -  \frac{2}{q} {\rm Tr} (\Psi^{\frac{q}{2}})  \right)  \ . 
\end{eqnarray}
 for $q$ such that $1/(p/2)+1/(q/2)=1$; and  $\Sb_{0+}$  denotes the set of symmetric positive semi-definite matrices:
$$\Sb_{0+} = \{ X \in \Sb: X \succeq 0 \}.$$
The relationship between \eqref{eq:gSPL} and \eqref{eq:Gssl2} can be clearly understood by minimizing \eqref{eq:gSPL} with respect to $\Psi$. Indeed,  using $\partial {\rm Tr}(A\Psi)/\partial \Psi = A^*$ and $\partial {\rm Tr}(\Psi)^{q/2}\partial \Psi =  \Psi^{q/2-1}$ \cite{petersen2008matrix},  we have 
\begin{eqnarray}\label{eq:psi}
\Psi= (Q^*A\Gamma A^*Q)^{\frac{1}{q/2-1}}
\end{eqnarray}
 and 
$$\min_{\Psi\in \Sb_{0+}} {\rm Tr}\left( (Q^*A\Gamma A^*Q)\Psi \right) -  \frac{2}{q} {\rm Tr} (\Psi^{\frac{q}{2}}) = \frac{1}{p} {\rm Tr} \left((Q^*A\Gamma A^*Q)^{\frac{p}{2}}\right)
$$ 
Here, $(Q^*A\Gamma A^*Q)^{\frac{1}{q/2-1}}$ should be understood as applying the power operation to the non-zero singular values of $Q^*A\Gamma A^*Q$ while retaining  zero singular values at zero.

Notice that $ {\mathfrak G}_{SPL}(\gammab,X, \Psi) $ is a surrogate function that majorizes ${\cal G}_{SPL}(X,\gammab) $.
Although like \eqref{eq:Gssl2},  \eqref{eq:gSPL} is not jointly convex with   respect to the different variables, the reason to prefer \eqref{eq:gSPL} over \eqref{eq:Gssl2} is that  \eqref{eq:gSPL} is convex with respect to each of  the variables  $\gammab, X$, and $\Psi$ with the other held constant, and we can obtain a closed-form expression in each step of alternating minimization.
Specifically, recall that the SPL penalty is given by
\begin{eqnarray}\label{eq:gSPL0}
 g_{SPL}(X) 
&=& 
\min_{\gammab \geq \zerob,\Psi \in \Sb_{0+} }{\mathfrak G}_{SPL}(\gammab, X,\Psi) 
\end{eqnarray}
where 
$${\mathfrak G}_{SPL}(\gammab, X,\Psi) = {\rm Tr}\left(X^*\Gamma^{-1}X\right) 
+   Np \left( {\rm Tr}\left( (Q^*A\Gamma A^*Q)\Psi \right) -  \frac{2}{q} {\rm Tr} (\Psi^{\frac{q}{2}})  \right)  .$$
Let $S$ denotes the non-zero support set of $X$. Using the KKT condition with respect to $\gammab$, we have
\begin{eqnarray} 
\frac{\partial{\mathfrak G}_{SPL}(\gammab, X,\Psi) }{\partial \gamma_i}  - \mu_i &=&   -\frac{\|\xb^i\|^2}{\gamma_i^2} + N \ab_i^*Q \Psi Q^*\ab_i 
-\mu_i= 0,\quad \forall ~i  \\
\mu_i \gamma_i &=&0, \quad 
\mu_i \geq  0, \quad
\gamma_i \geq 0, \quad \forall ~i  
\end{eqnarray}
which leads  to $\gamma_i =0$ for $i \notin S$, whereas for $i \in S$ 
$$\gamma_i  \frac{\partial {\mathfrak G}_{SPL}}{\partial \gamma_i}  = 0 = -  \frac{\|\xb^i\|^2}{\gamma_i}+  N \gamma_i \ab_i^*Q(Q^*A \Gamma A^*Q)^{\frac{1}{q/2-1}} Q^*\ab_i \  .$$
Hence, 
${\rm Tr}\left(X^*\Gamma^{-1}X\right)  = N {\rm Tr} \left( (Q^*A \Gamma A^*Q)^{1+\frac{1}{q/2-1}} \right)= N {\rm Tr} \left( (Q^*A \Gamma A^*Q)^{\frac{p}{2}} \right)$ and we have
$$ g_{SPL}(X)= 2N {\rm Tr} \left(\left( Q^*A \Gamma A^*Q \right)^{\frac{p}{2}})\right) = 2N {\rm Tr} |Q^*A \Gamma^{\frac{1}{2}} |^{p} .$$
This implies that at the KKT point, the SPL penalty has  cost function values  equivalent to the Schatten-$p$ quasi-norm rank penalty for $Q^*A \Gamma^{\frac{1}{2}}$.

\section{The SPL Algorithm}

\subsection{Alternating Minimization Algorithm}
So far, we have analyzed the global minimizer for the noiseless SPL algorithm.  For noisy measurement, we propose the following cost function:
\begin{eqnarray}
 \min_X  \|Y-AX\|_F^2+ \lambda g_{SPL}(X)  \  . \label{eq:noisySPL}
\end{eqnarray}
By letting 
$\lambda\rightarrow 0$, the solution of \eqref{eq:noisySPL} becomes  a solution of \eqref{eq:noiselessSPL} when ${\rm rank}(A\Gamma^{1/2})=m$ since then the constraint is automatically satisfied as follows:
$$\lim_{\lambda \rightarrow 0} A X(\lambda) = \lim_{\lambda \rightarrow 0} A  \Gamma A^*(\lambda I+A \Gamma A^* )^{-1} Y   = Y$$
Similar equivalence can be hold for ${\rm rank}(A\Gamma^{1/2})<m$  if $Y \in R(A\Gamma^{1/2})$.
Therefore, rather than dealing with Eqs.~\eqref{eq:noiselessSPL}  and \eqref{eq:noisySPL} separately,  we use \eqref{eq:noisySPL} and the limiting argument
to discuss a noiseless SPL optimization problem.

Then, using \eqref{eq:gSPL0},  a noisy SPL formulation can be written as
\begin{eqnarray}\label{eq:SPLform}
 \min\limits_{X, \gammab\geq \zerob,\Psi\in \Sb_{0+}}   &&C (X,\gammab,\Psi) 
 \end{eqnarray}
 where the augmented cost function is given by
 \begin{eqnarray}\label{eq:augcost}
 C(X,\gammab,\Psi) &=&  \|Y-AX\|_F^2+ \lambda \left\{{\rm Tr}\left(X^*\Gamma^{-1}X\right) 
+   N p \left[{\rm Tr}\left( (Q^*A\Gamma A^*Q)\Psi \right) -   \frac{2}{q} {\rm Tr} (\Psi^{\frac{q}{2}}) \right]\right\}  \ .
\end{eqnarray}
While $C(X,\gammab,\Psi)$ is not convex for all these variables simultaneously due to the presence of the bi-convex terms  ${\rm Tr}\left(X^*\Gamma^{-1}X\right)$ and ${\rm Tr}\left( (Q^*A\Gamma A^*Q)\Psi \right)$,    it is convex with respect to each variable $X, \gammab$ and $\Psi$ separately. Indeed, this is a typical example of the d.c. algorithm (DCA) for the difference of convex functions programming \cite{tao1997convex,tao1998dc}, and the alternating minimization algorithm converges to a {\em local} minimizer or a critical point.

Specifically,   a critical solution should satisfy the following first order Karush-Kuhn-Tucker (KKT) necessary conditions \cite{ChZa96}:
\begin{eqnarray} 
\frac{\partial  C(X,\gammab,\Psi)}{\partial X} &=& -2A^*(Y-AX)+ \lambda\Gamma^{-1}X = \zerob \label{eq:xiter}\\
\frac{\partial  C(X,\gammab,\Psi)}{\partial \Psi} &=& Np (Q^*A\Gamma A^*Q) - Np\Psi^{q/2-1}= 0 \label{eq:psiiter} \\
\frac{\partial C(X,\gammab,\Psi)}{\partial \gamma_i}  - \mu_i &=&   -\frac{ \sum_j |x_{ij}|^2}{\gamma_i^2} + N \ab_i^*Q \Psi Q^*\ab_i 
-\mu_i= 0,\quad \forall ~i  \label{eq;gammaiter}\\
\mu_i \gamma_i &=&0, \quad 
\mu_i \geq  0, \quad
\gamma_i \geq 0, \quad \forall ~i \label{eq:lambda}
\end{eqnarray}
This leads us to the following fixed point iterations:

\subsubsection{Minimization with respect to $X$}

For a given estimate $\gammab^{(t)}$,   \eqref{eq:xiter} yields a closed form solution for $X^{(t+1)}$:
$$X^{(t+1)}=\Gamma^{(t)} A^*(\lambda I+A\Gamma^{(t)}A^*)^{-1}Y,\quad \Gamma^{(t)}=\textrm{diag}(\gammab^{(t)}).$$

\subsubsection{Determination of $\Psi$}

For a given estimate $\gammab^{(t)}$,  using \eqref{eq:psiiter},  we can  find a closed-form solution for $\Psi^{(t)}$: i.e.
$\Psi^{(t)}= (Q^*A\Gamma^{(t)} A^*Q)^{\frac{1}{q/2-1}} \  .$

\subsubsection{Estimation of $\gammab$}

For a given $\Psi^{(t)}$ and $X^{(t)}$,  using Eqs.~\eqref{eq;gammaiter} and \eqref{eq:lambda}, we have
\begin{eqnarray}\label{eq:kkt_g}
\gamma_i\frac{\partial C(X^{(t)},\gammab,\Psi^{(t)})}{\partial \gamma_i} = \gamma_i\left(-\frac{  \| \xb^{(t)i} \|^2}{\gamma_i^2} + N \ab_i^*Q \Psi^{(t)}Q^*\ab_i \right)= 0\ .
\end{eqnarray}
Here, if  $\ab_i^*Q \Psi^{(t)}Q^*\ab_i \neq 0$, we have
the following update equation:
\begin{eqnarray}\label{eq:gamma_ssl}
\gamma_i^{(t)} = \left(\frac{\frac{1}{N} \| \xb^{(t)i} \|^2}{\ab_i^*Q\Psi^{(t)}Q^*\ab_i}\right)^{\frac{1}{2}}  = \left(\frac{\frac{1}{N} \| \xb^{(t)i} \|^2}{\ab_i^*Q(Q^*A\Gamma^{(t)}A^*Q)^{\frac{1}{q/2-1}}Q^*\ab_i}\right)^{\frac{1}{2}}\  .
\end{eqnarray}

Note that the SPL updates appears similar to those of M-SBL except the $\gammab$ update by \eqref{eq:gamma_ssl}, which is now modified based on subspace geometry. This is the main ingredient for the performance improvement of SPL over M-SBL. In the following, we further discuss  several important properties of the SPL penalty.

\subsection{Properties  of  the SPL Penalty}

An interesting case occurs when  $\ab_i^*Q \Psi^{(t)}Q^*\ab_i = 0$.
   In this case,  based on \eqref{eq:kkt_g}, we have the following two observations: 1) $\gamma_i^{(t)}\rightarrow \infty$ when $ \| \xb^{(t)i} \|^2\neq 0$; and 2)  $\gamma_i^{(t)}$ can be an arbitrary positive number $C$ when  $ \| \xb^{(t)i} \|^2= 0$ since  the equality in \eqref{eq:kkt_g}  is satisfied regardless of the choice of $C$. Therefore, 
we  define  the following $\gamma_i$ update\footnote{In a practical implementation, a tolerance around 0 and finites values for $\gamma_i^{(t)}$ have to be used.}:
\begin{eqnarray}\label{eq:gammaupdate}
\gamma_i^{(t)} = \left\{ \begin{array}{ll} 
\infty ,& \mbox{if  $\ab_i^*Q \Psi^{(t)}Q^*\ab_i =0$ and $ \| \xb^{(t)i} \|^2\neq 0$} \\
C\gg 0,& \mbox{if  $\ab_i^*Q \Psi^{(t)}Q^*\ab_i =0$ and $ \| \xb^{(t)i} \|^2= 0$} \\
\left(\frac{\frac{1}{N} \| \xb^{(t)i} \|^2}{\ab_i^*Q\Psi^{(t)}Q^*\ab_i}\right)^{\frac{1}{2}}
& \mbox{if  $\ab_i^*Q \Psi^{(t)}Q^*\ab_i \neq 0$ } 
 \end{array} \right.
\end{eqnarray}
Thanks to \eqref{eq:gammaupdate},  even if $ \| \xb^{(t)i} \|^2$ becomes erroneously zero during the iterations, there is a possibility, when $ \Psi^{(t)}Q^* \ab_i =0$, for $\gamma_i^{(t)}$ to become nonzero; hence, the corresponding row of $X^{(t)}$ can become nonzero once
$\gamma_i$ turns into nonzero.
%
%
 Note that this is very different from M-SBL, since in \eqref{eq:msbl_gamma} the denominator term cannot be zero even under the most relaxed RIP constraint $\delta_{k+1}^L<1$, so the condition
$  \| \xb^{(t)i} \|^2=0$ will set  the corresponding $\gamma_i^{(k)}$ to zero. Therefore,  in M-SBL, once a row of $X^{(t)}$ is set to zero in error, it will stay  zero for all subsequent iterations and the algorithm is unable to recover from this error.

Second, it is important note that since $1/(p/2)+1/(q/2)=1$, we have $\lim_{q\rightarrow 0 } p = \lim_{q\rightarrow 0 }\frac{2}{1-2/q} =0$; so
\begin{eqnarray}
\lim_{q\rightarrow 0} g_{SPL}(X)=  \lim_{q\rightarrow 0} 2N {\rm Tr} |Q^*A \Gamma^{\frac{1}{2}} |^{p}  = 2N {\rm Rank}(Q^*A_I),
\end{eqnarray}
where $I$ denotes the index set of non-zero diagonal elements of $\Gamma$. Hence, in this case, the SPL algorithm with $p\rightarrow 0$ is the algorithm that directly minimizes the rank
of $Q^*A_I$. In this case, the corresponding update rule is given by
\begin{eqnarray}
X^{(t+1)}=\Gamma^{(t)} A^*(\lambda I+A\Gamma^{(t)}A^*)^{-1}Y &,&
\gamma_i^{(t)} = \left(\frac{\frac{1}{N} \| \xb^{(t)i} \|^2}{\ab_i^*Q(Q^*A\Gamma^{(t)}A^*Q)^{-1}Q^*\ab_i}\right)^{\frac{1}{2}}  \label{eq:spl_gamma}\ .
\end{eqnarray}
The main technical challenge is, however, that for $p=q=0$ the cost function \eqref{eq:augcost} is not well-defined. Therefore, the aforementioned interpretation of the SPL should be understood as an asymptotic result such that $p$ and $q$ approach  zero, but are not exactly zero.

Next, as a by product of Theorem~\ref{thm:spark}, the SPL algorithm is computationally more efficient than M-SBL. Note that  the computational bottleneck of M-SBL (or SPL) is due to the 
the inversion of $A\Gamma^{(t)}A^*$  (or $Q^*A\Gamma^{(t)}A^*Q$, respectively). Specifically, unlike the $X^{(t)}$ update step that can be done using the conjugate gradient (CG) algorithm,  the matrix inversion cannot be performed using CG and usually is performed using the singular value decomposition (SVD).  Now, note that the size of  matrix $Q^*A\Gamma^{(t)}A^*Q$ in SPL is $(m-r) \times (m-r)$ compared to $m \times m$ for $A\Gamma^{(t)}A^*$, which reduces the cost of matrix inversion for SPL  compared to M-SBL. In particular, for the case of MUSIC where 
$m=k+1$ and $r=k$,   matrix inversion is not necessary for SPL whereas M-SBL still requires the $m\times m$ matrix inversion.

Finally, note that the hyper-parameter $\gammab$ is closely related to  spectral estimation. For example,   for the case of MUSIC where 
$m=k+1$ and $r=k$, the term $(Q^*A\Gamma^{(t)}A^*Q)$ in \eqref{eq:spl_gamma} reduces a  scalar and we have
\begin{eqnarray}
\gamma_i^{(t)}  &=& \left(\frac{\frac{1}{N} \| \xb^{(t)i} \|^2}{\ab_i^*Q(Q^*A\Gamma^{(t)}A^*Q)^{-1}Q^*\ab_i}\right)^{\frac{1}{2}} \\
&=&  \frac{1}{\sqrt{\ab_i^*Q Q^*\ab_i} } \times \frac{  \| \xb^{(t)i} \|}{\sqrt{N (Q^*A\Gamma^{(t)}A^*Q)}}
\end{eqnarray}
where the first term is the MUSIC spectrum and the second term is related to the magnitude of the $i$-th row of $X^{(t)}$.
Hence,  in the case of full-row rank $X$ (i.e., the MUSIC case),  SPL can be regarded as an algorithm that initialises he non-zero support estimation using a spectral estimation technique, followed by alternating modification using the data fidelity matching criterion.

\section{NUMERICAL RESULTS}
\label{sec:results}

In this section, we perform extensive numerical experiments to validate the proposed algorithm  under various experimental conditions, and compare it with respect to  existing joint sparse recovery algorithms.
In particular, we are interested in the SPL algorithm in the asymptotic region of $p \rightarrow 0$ since it directly minimises the rank of $Q^*A_I$.

%
The elements of a sensing matrix $A$ were generated  from a Gaussian distribution with zero mean and variance of $1/m$, and then each column of $A$ was normalized to have an unit norm.  An unknown signal $X_*$ with ${\rm rank}(X_*)=r\leq k$ was generated using the same procedure as in \cite{Lee2010SAMUSIC}. Specifically, we  randomly generated a support $I$, and then the corresponding nonzero signal components were obtained by
\begin{equation}\label{eq:source}
X_*^I=U\Sigma V \ ,
\end{equation} where $U \in \Re^{k\times r}$  was set to random orthonormal columns, and $\Sigma={\rm diag}([\sigma_i]_{i=1}^r)$  is a diagonal matrix whose $i$-th element is given by
\begin{eqnarray}\label{eq:tau}
\sigma_i = \tau^i,\quad  0<\tau <1,
\end{eqnarray}
 and $V \in \Re^{r\times N}$ was generated using Gaussian random distribution with zero mean and variance of $1/N$. After generating noiseless data,  we added zero mean white Gaussian noise. We declared success if an estimated support from a certain algorithm was the same as a true ${\rm supp}X$.
 
As the proposed algorithm does not require a prior knowledge of the sparsity level, we need to define a stoping criterion.  Here, the stopping criterion is defined by monitoring the normalized change in the variable $\gammab$:
$$ \frac{\| \gammab^{(t)} - \gammab^{(t-1)}\|_2}{\|\gammab^{(t)}\|_2} < 10^{-3}\ .$$
From our experiments, usually 20-30 iterations are required for SPL to converge. 

\subsection{Local Minima Property}

We first perform experiments to confirm that SPL produces the true solution under milder conditions than M-SBL. To show this,  using $k$-sparse signal $X_*$ generated by \eqref{eq:source} with $\tau=1$, we  produced measurements $Y=AX_*$ such that $r:={\rm rank}(Y)\leq k$.   Then,  we initialized both algorithms with  $X^{(0)}$ that satisfies the following:
 \begin{eqnarray}\label{eq:local_sim}
 Y=AX^{(0)},\quad   \|X^{(0)}\|_0=m, \quad |{\rm supp} X_* \cap  {\rm supp} X^{(0)}| = s \ , 
 \end{eqnarray}
 where $s<k-r$, $s=k-r$ and $s>k-r$, respectively. Note that the initialization corresponds to a local minimiser and we are interested in confirming that SPL can escape from the local minimizers thanks to the update in \eqref{eq:gammaupdate}.  Recall that it is difficult for M-SBL to avoid this type of local minimizers since $X^{(0)}$ has zeros rows at 
the $i$-th row where $i \in S\setminus  {\rm supp} X^{(0)}$and the M-SBL update rule  in \eqref{eq:msbl_gamma} cannot make  the corresponding $\gamma_i$ nonzero in the subsequent iterations.
 
   Figs.~\ref{fig:localminimizer}(a)-(c) illustrate the perfect recovery ratio from the initialization using SPL and M-SBL at various SNR conditions for (a)  $s=k-r=3$, (b) $s=2< k-r$, and (c) $s=5>k-r$, respectively. 
   The results clearly demonstrate that SPL  finds the global minimizer nearly perfectly, whereas M-SBL fails most of the time. This clearly confirms the theoretical advantages   for SPL.

\begin{figure}[htbp]
\centering{\includegraphics[width=5cm]{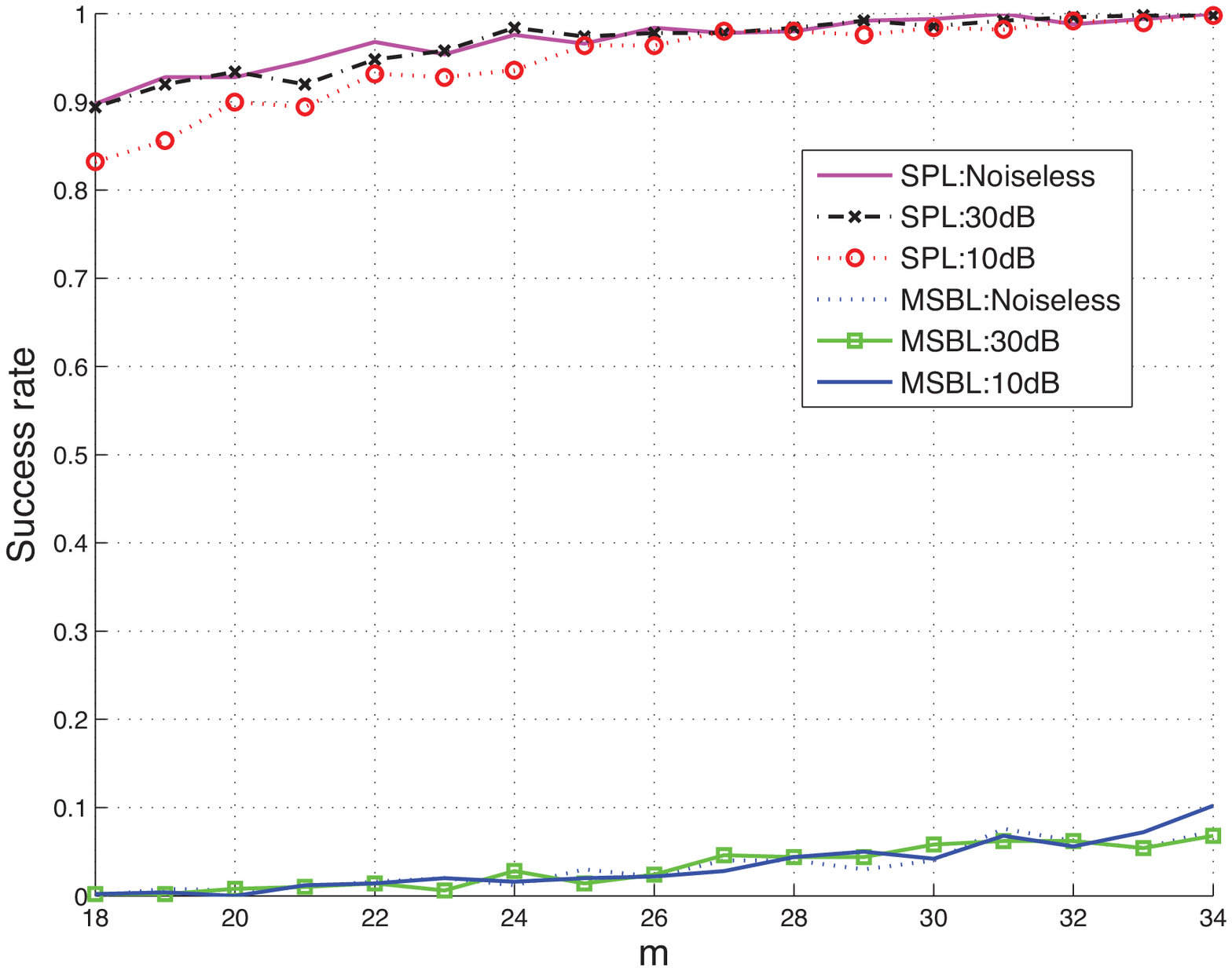}\includegraphics[width=5cm]{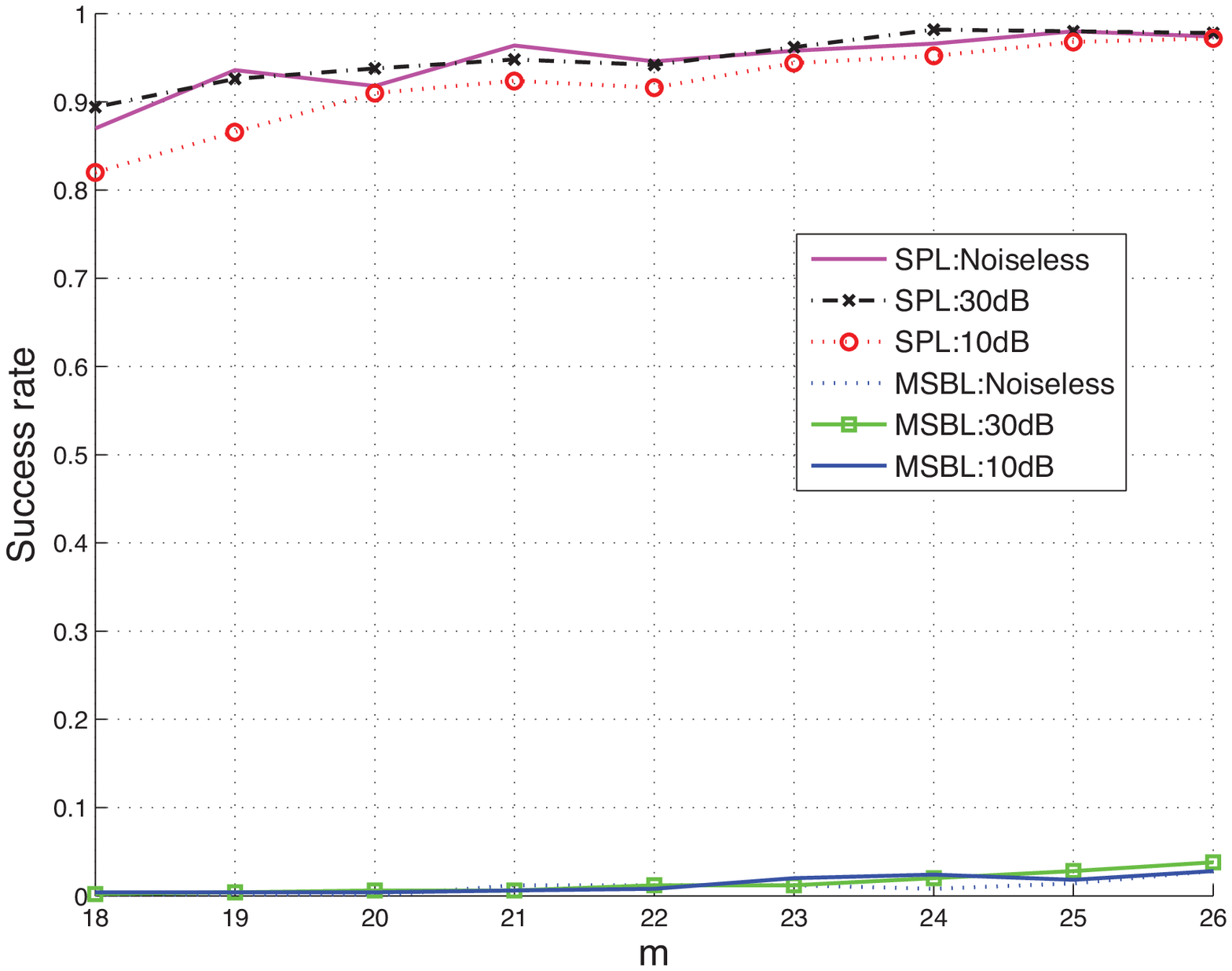}\includegraphics[width=5cm]{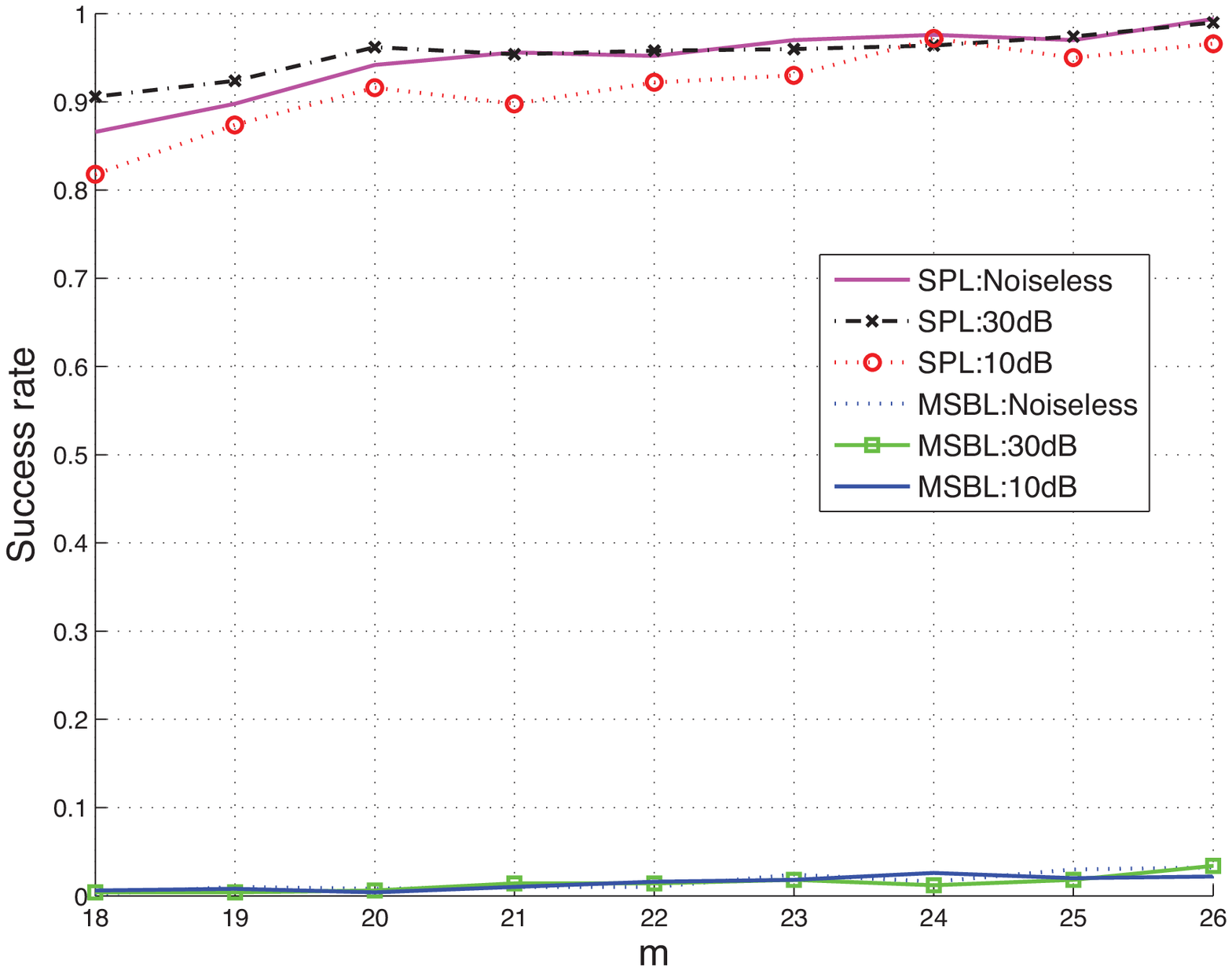}}\\
\centering{\mbox{(a)}\hspace{5cm}\mbox{(b)}\hspace{5cm}\mbox{(c)}}
\caption{Perfect recovery ratio from initialization using local minimizer that satisfy \eqref{eq:local_sim}. The results are averaged after 500 runs and the simulation parameters are: $k=10,r=7, n=128$ and $N=64$. (a)  $s=k-r=3$, (b) $s=2< k-r$, and (c) $s=5>k-r$, respectively.}
\label{fig:localminimizer}
\end{figure}

\subsection{Comparison with Other State-of-Art Algorithms}

\begin{figure}[hbpt]
\centering{\includegraphics[width=6cm]{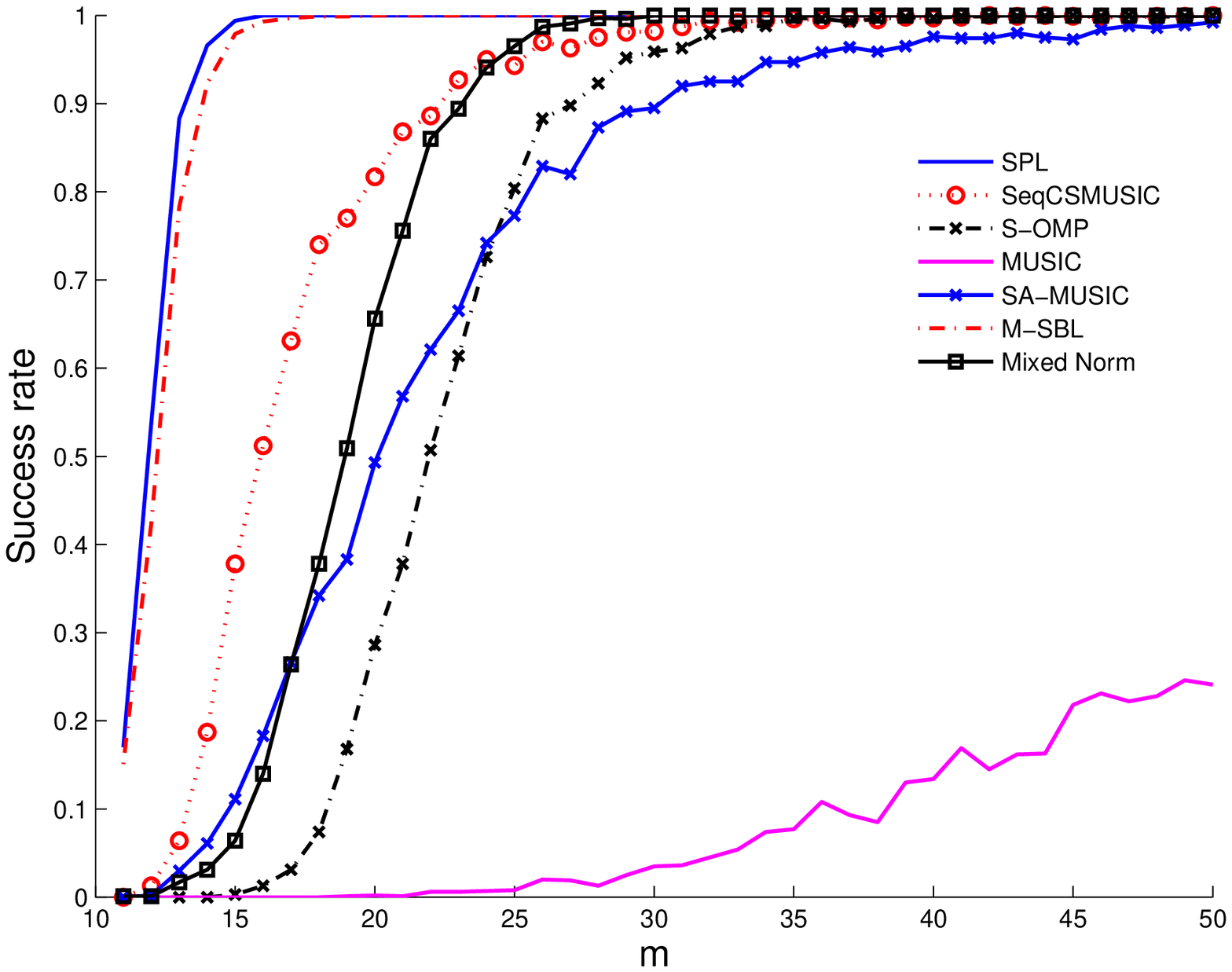}\includegraphics[width=6cm]{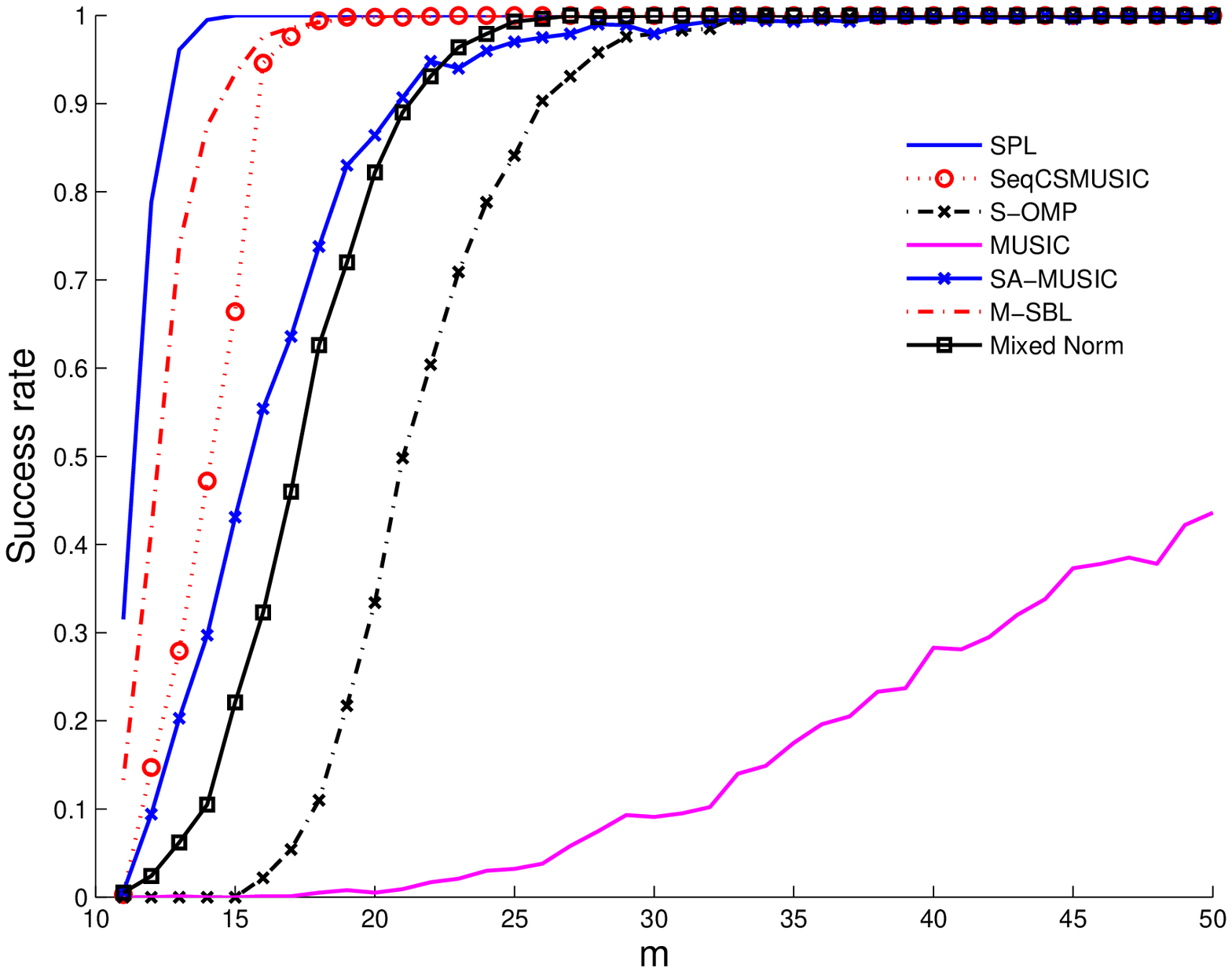}}\\
\centering{\mbox{(a)}\hspace{6cm}\mbox{(b)}}\vspace*{0.5cm}
\centering{\includegraphics[width=6cm]{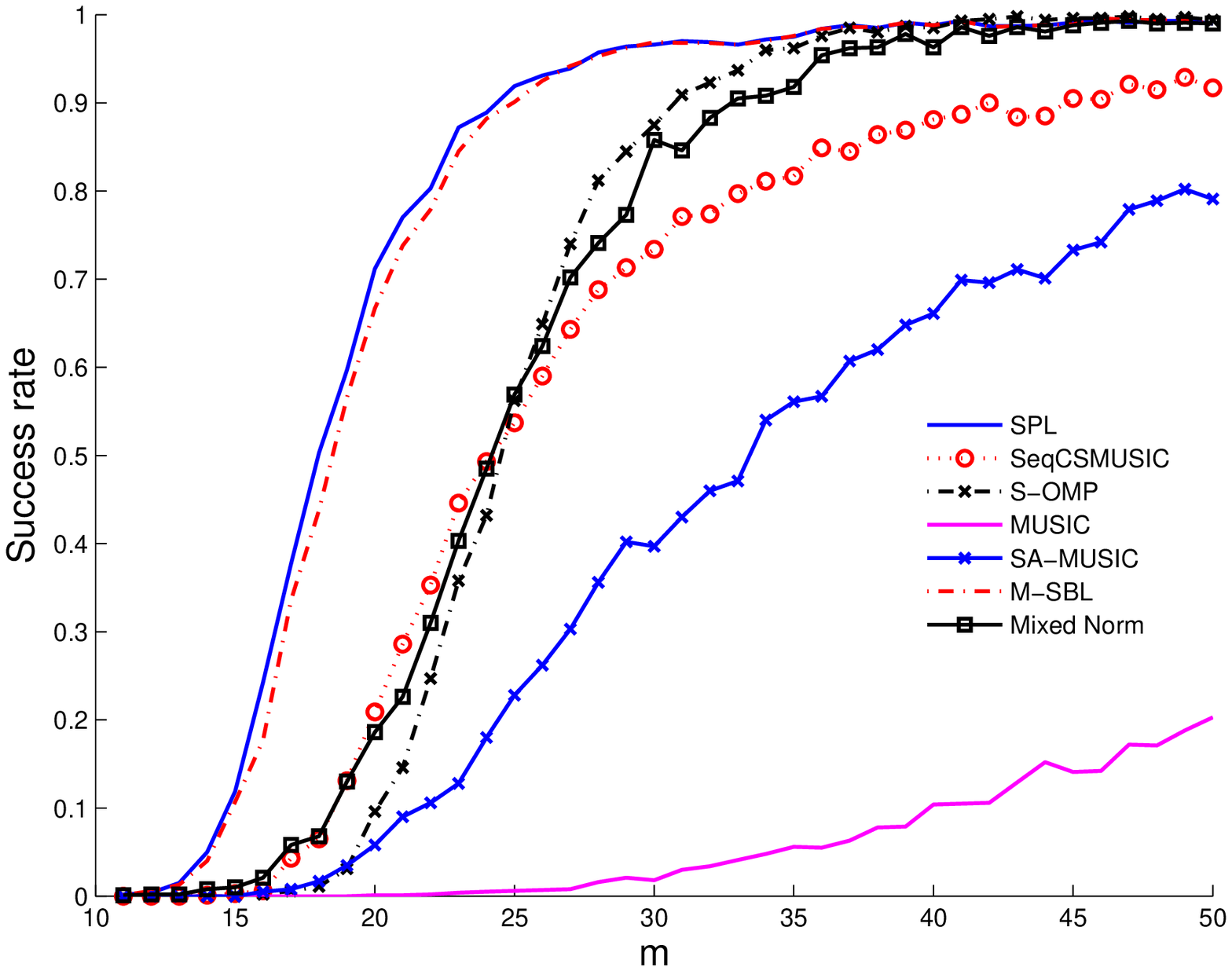}\includegraphics[width=6cm]{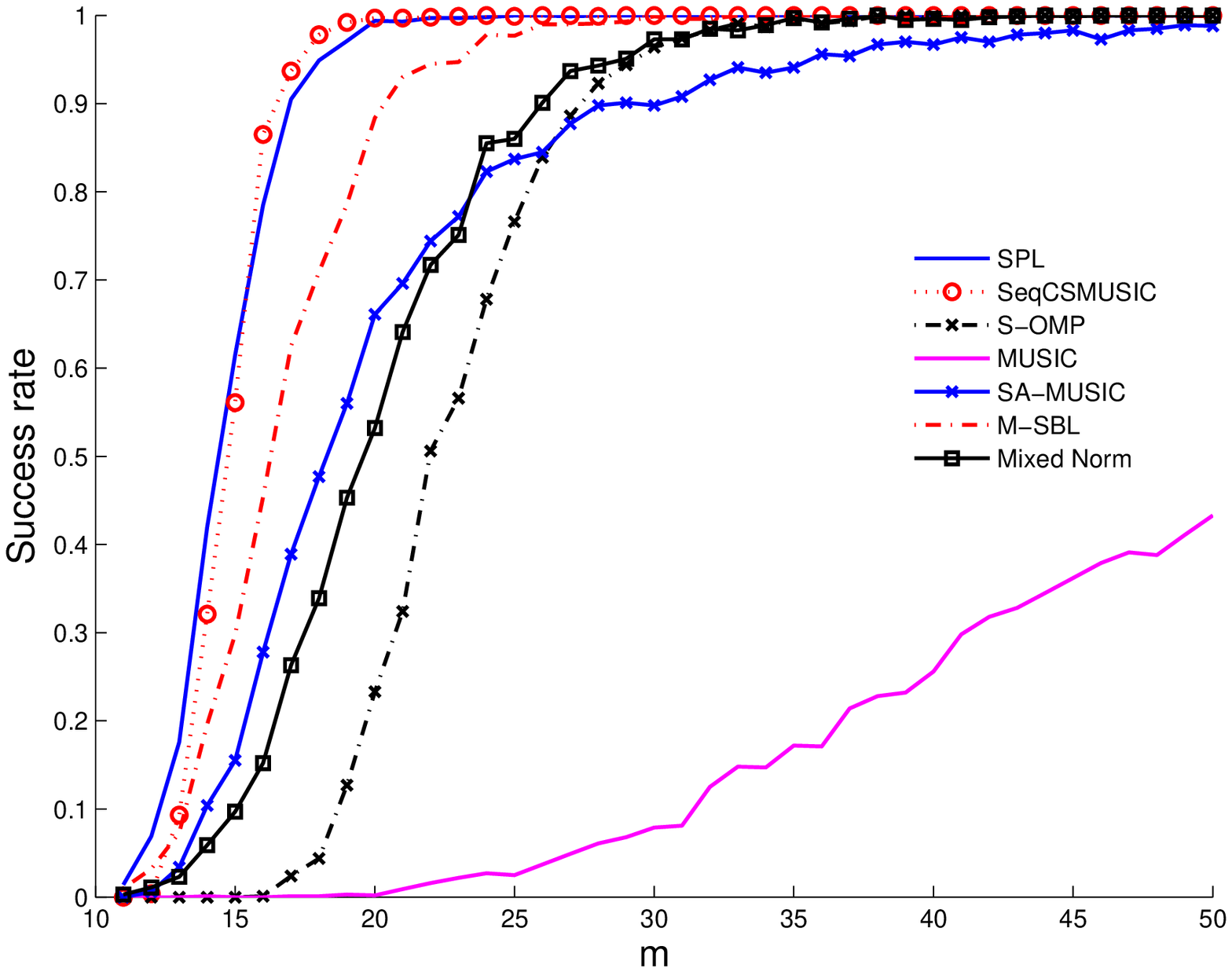}}\\
\centering{\mbox{(c)}\hspace{6cm}\mbox{(d)}}\vspace*{0.5cm}
\centering{\includegraphics[width=6cm]{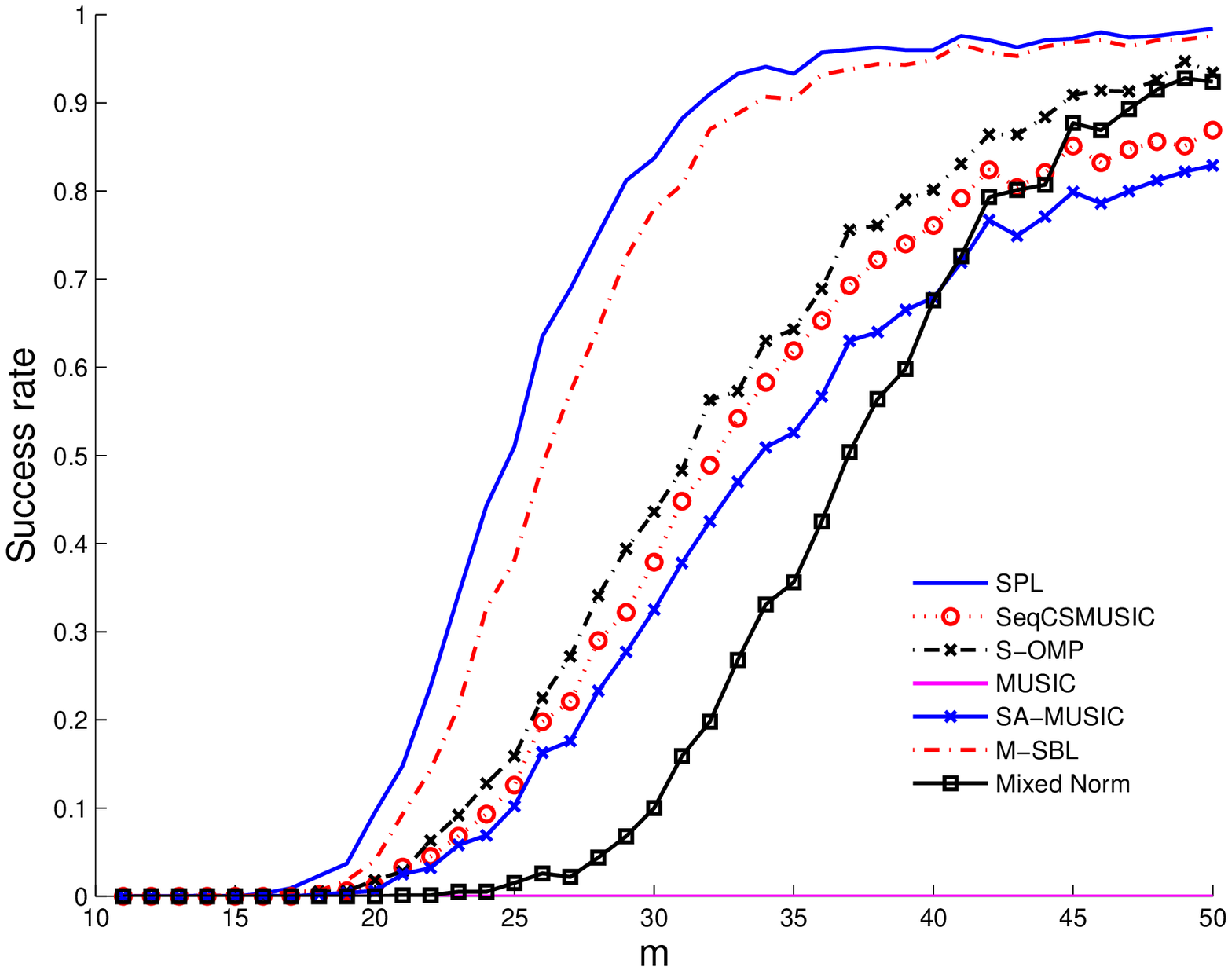}\includegraphics[width=6cm]{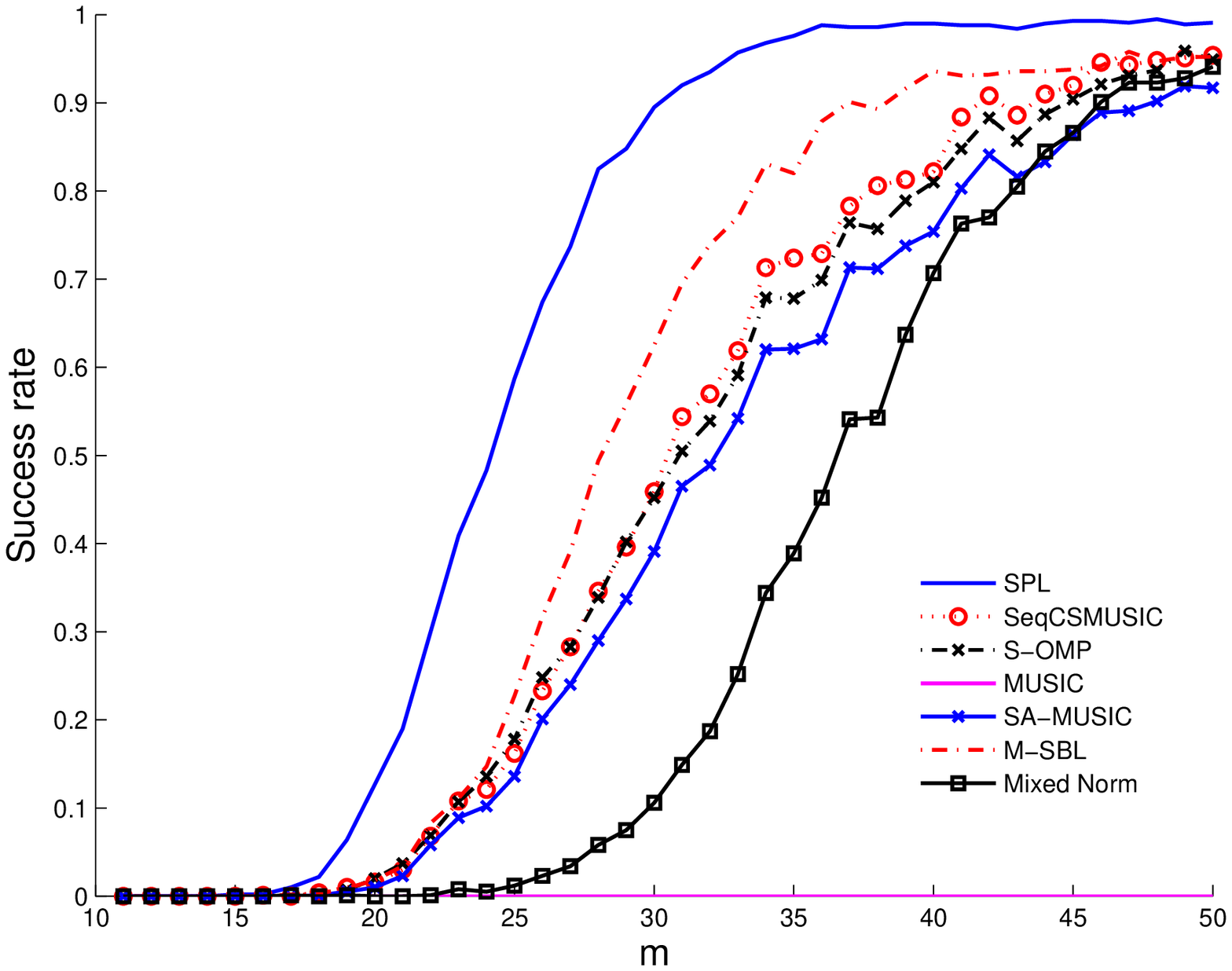}}\\
\centering{\mbox{(e)}\hspace{6cm}\mbox{(f)}}
\caption{Performance of various joint sparse recovery algorithms at $n=128, k=10, r=6$ when
(a) $SNR=30dB, N=16, \tau=1$, (b) $SNR=30dB, N=256, \tau=1$, 
(c) $SNR=10dB, N=16, \tau=1$, (d) $SNR=10dB, N=256, \tau=1$, (e) $SNR=30dB, N=16, \tau=0.1$, (f) $SNR=30dB, N=256, \tau=0.1$, respectively.}
\label{fig:variousMMV}
\end{figure}

To compare the proposed algorithm with various state-of-art joint sparse recovery methods,  the recovery rates of various state-of-art joint sparse recovery algorithms such as   MUSIC, S-OMP, SA-MUSIC,   sequential CS-MUSIC,   M-SBL, and  the $l_1/l_2$ mixed norm approach 
are plotted in Fig.~\ref{fig:variousMMV} along with  those of  SPL.  
Among the various implementation of mixed norm approaches, we used high performance  SGPL1 software \cite{van2008probing}, which can be downloaded from  ${\rm http://www.cs.ubc.ca/labs/scl/spgl1/}$. 
Since M-SBL, the mixed norm approach,  as well as SPL do not provide an exact $k$-sparse solution, we used the support for the largest $k$ coefficients as a  support estimate in calculating the perfect recovery ratio. For MUSIC, S-OMP, SA-MUSIC,   sequential CS-MUSIC, we assume that $k$ is known.
For subspace based algorithms such as MUSIC, SA-MUSIC, sequential CS-MUSIC as well as SPL, we determine the signal subspace using the following criterion
$$ \max_{i\in \{1,\cdots,m\}} \frac{\sigma_i -\sigma_{i+1}}{\sigma_{i}-\sigma_m} > 0.1,$$
where $\sigma_1\geq \sigma_2 \geq \cdots \geq \sigma_m$ denotes the singular values of $YY^*$.  A theoretical motivation for such subspace determination is given in \cite{Lee2010SAMUSIC}.
Here, the success rates were averaged over $1000$ experiments. 
The simulation parameters were as follows: $m\in\{1,2,\ldots,50\}$, $n=128, k=8, r=5$, $SNR=30dB, 10dB$, and $N\in\{32,128\}$, respectively. 
Figs.~\ref{fig:variousMMV}(a)-(d) illustrates the comparison results under various snapshot number and SNR conditions.
Note that SPL  consistently outperforms all other algorithms at various  snapshots numbers. In particular,  the gain increases with increasing number of snapshots,  since it provides better subspace estimation.  Also, note that  SPL  consistently outperforms  M-SBL at all SNR ranges.  
Figs.~\ref{fig:variousMMV}(e)(f) illustrates that SPL significantly outperforms M-SBL when $X$ is badly conditioned. Moreover, as the subspace estimation becomes accurate with increasing $N$, the performance gain becomes more significant. 


\begin{figure}[!tbp]
\centering{\includegraphics[width=6cm]{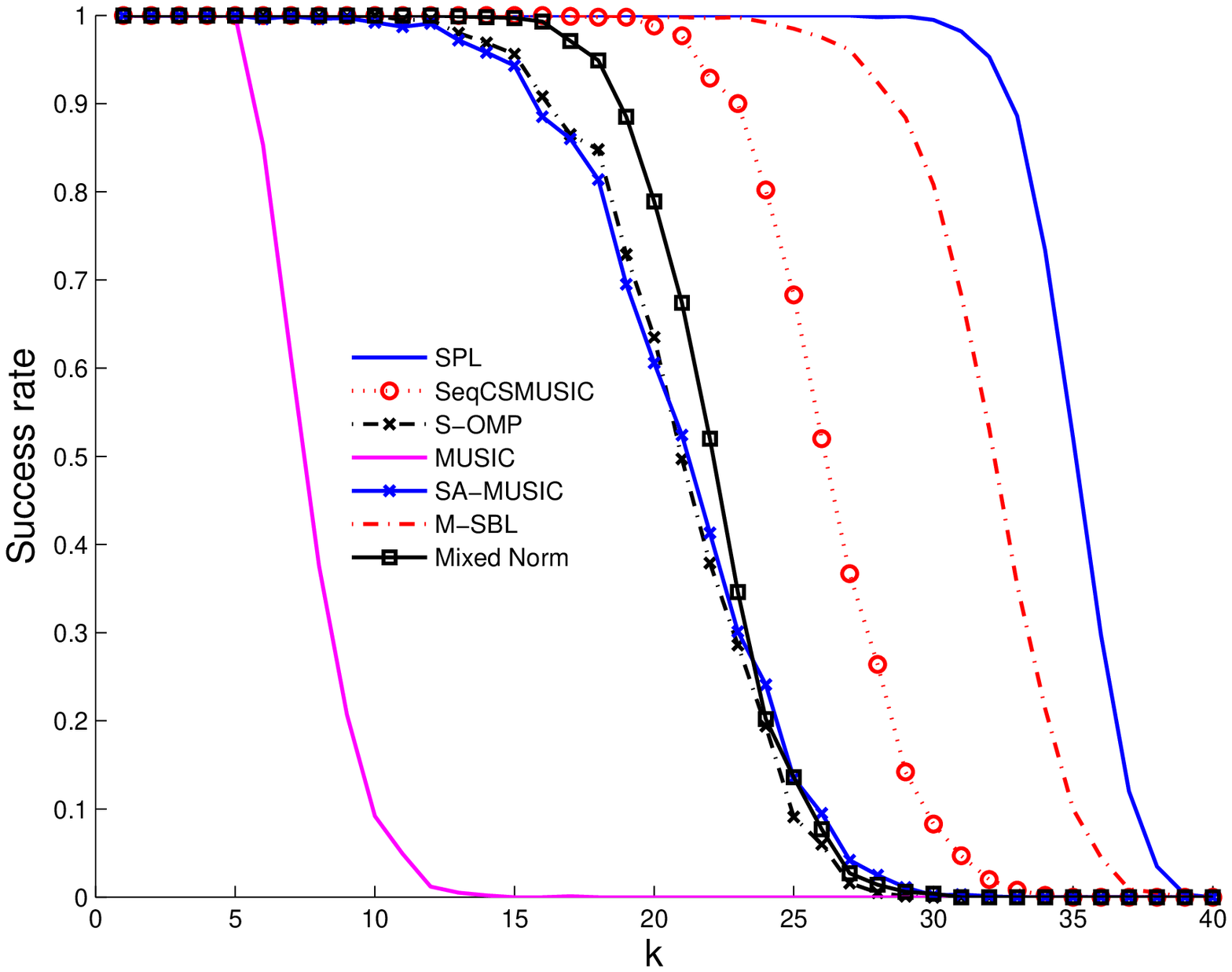}\includegraphics[width=6cm]{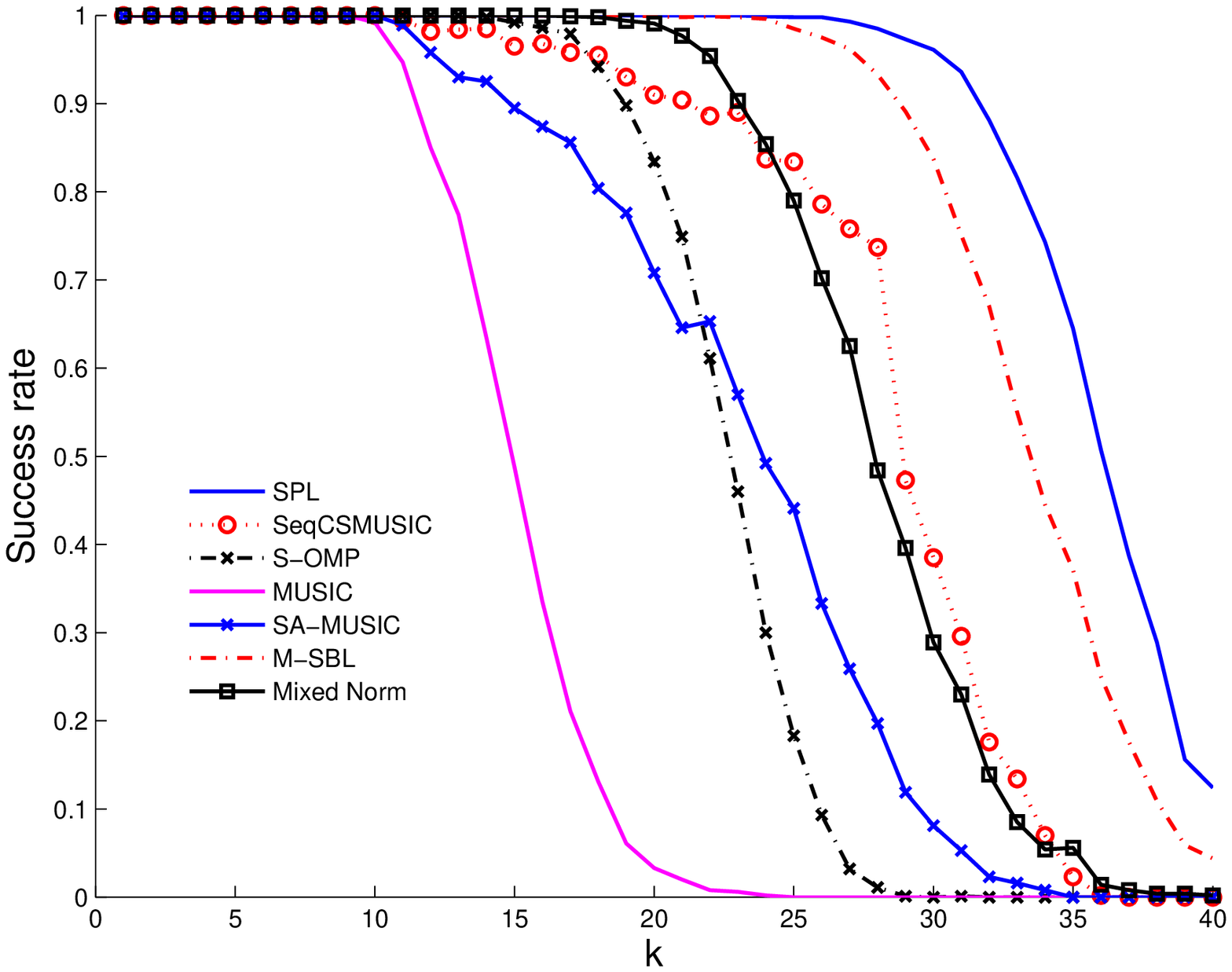}\includegraphics[width=6cm]{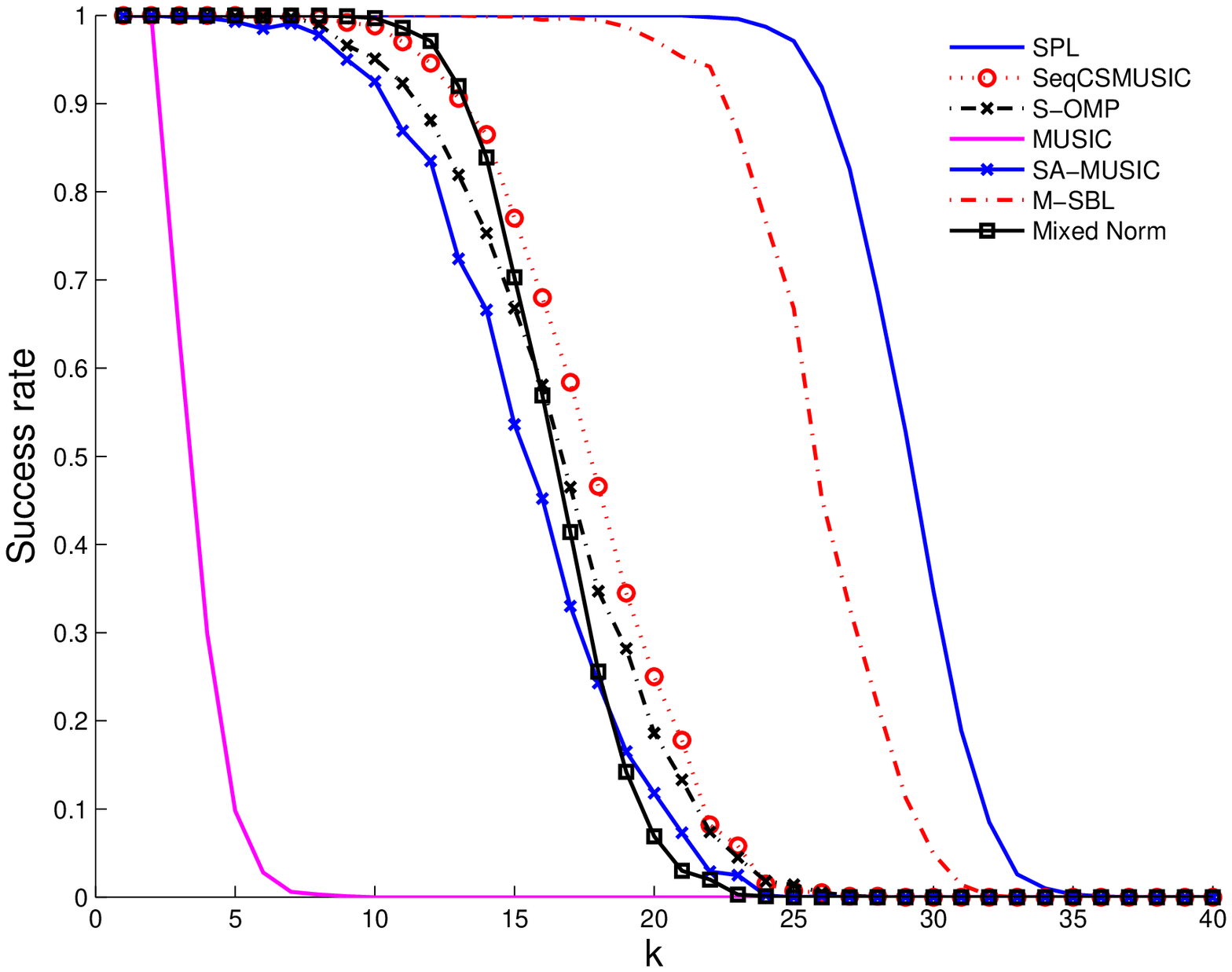}}\\
\centering{\mbox{(a)}\hspace{6cm}\mbox{(b)}\hspace{6cm}\mbox{(c)}}
\caption{Performance of various joint sparse recovery algorithms for varying sparsity level at $N=256$.  The simulation parameters are (a) $m=40, r=5, \tau=1$ and SNR=30dB, and  (b) $m=40, r=12, \tau=1$ and SNR=10dB,  and (c) $m=40,r=15,\tau=0.5$ and SNR=30dB, respectively.}
\label{fig:variousMMV_k}
\end{figure}

Figs.~\ref{fig:variousMMV_k}(a)(b)(c) compares the performance of various MMV algorithm by varying the sparsity level. Here,  $m$ and  ${\rm rank}(Y)$ are fixed and the sparsity levels changes, and we calculated the perfect reconstruction ratio.  Again, SPL outperforms all existing methods for various SNR and conditions numbers.

\subsection{Fourier Measurements Cases}

Fig.~\ref{fig:fourier} illustrates the results of the comparison when the measurement are from Fourier sensing matrix.  Similar to Gaussian sensing matrix,  consistent improvement of SPL over M-SBL and  other algorithms under various conditions have been observed.

\begin{figure}[hbpt]
\centering{\includegraphics[width=7cm]{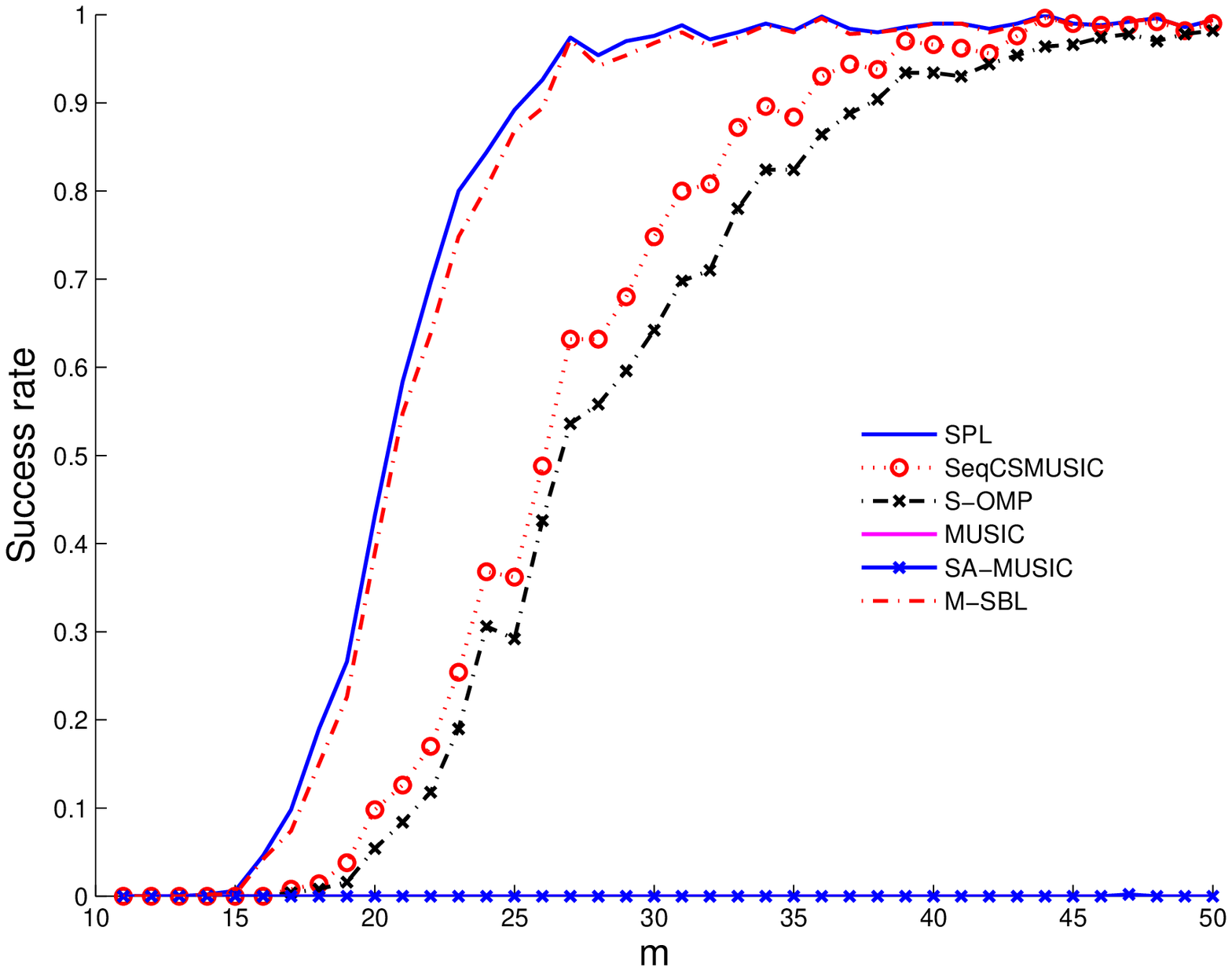}\includegraphics[width=7cm]{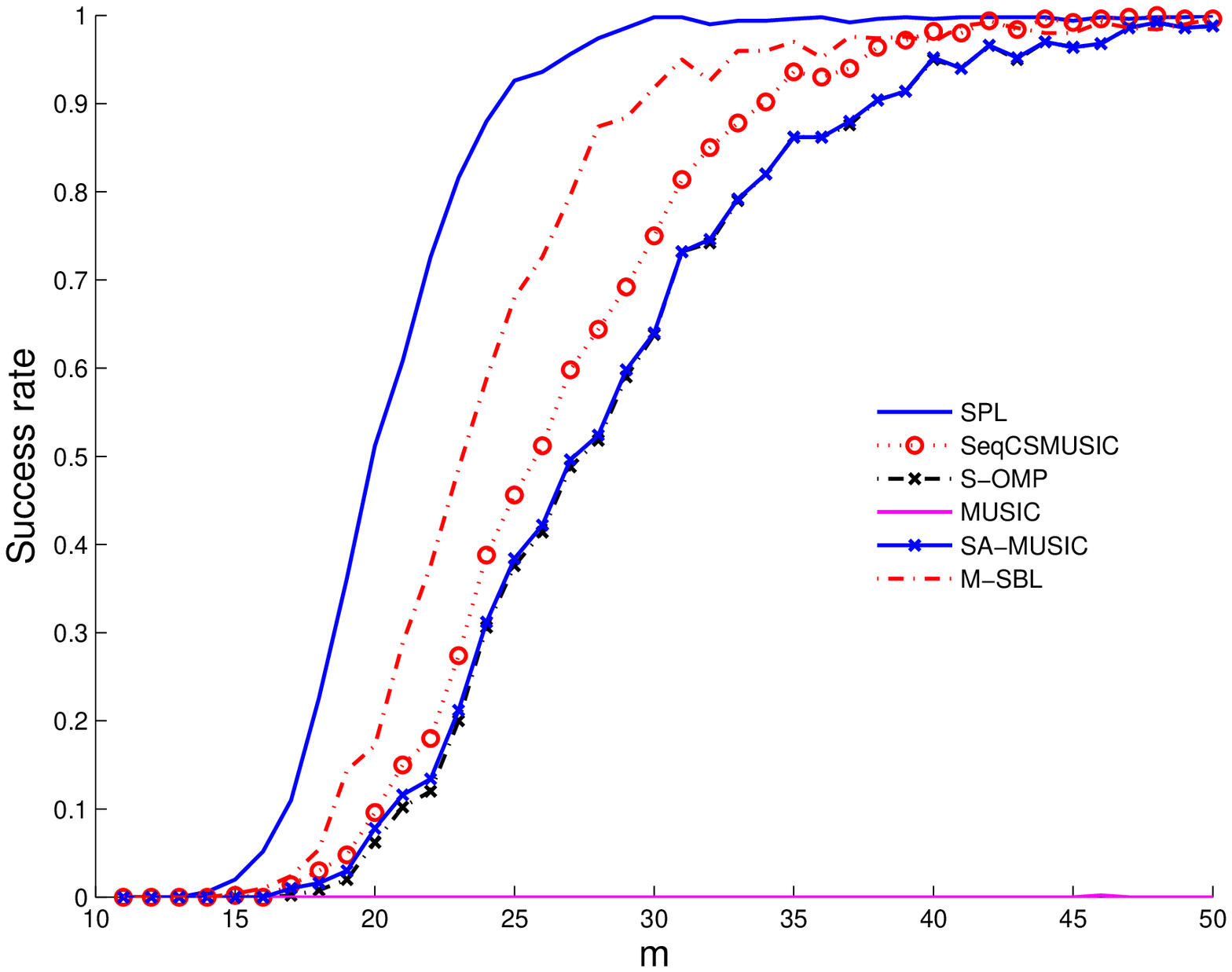}}\\
\centering{\mbox{(a)}\hspace{7cm}\mbox{(b)}}\vspace*{0.5cm}
\centering{\includegraphics[width=7cm]{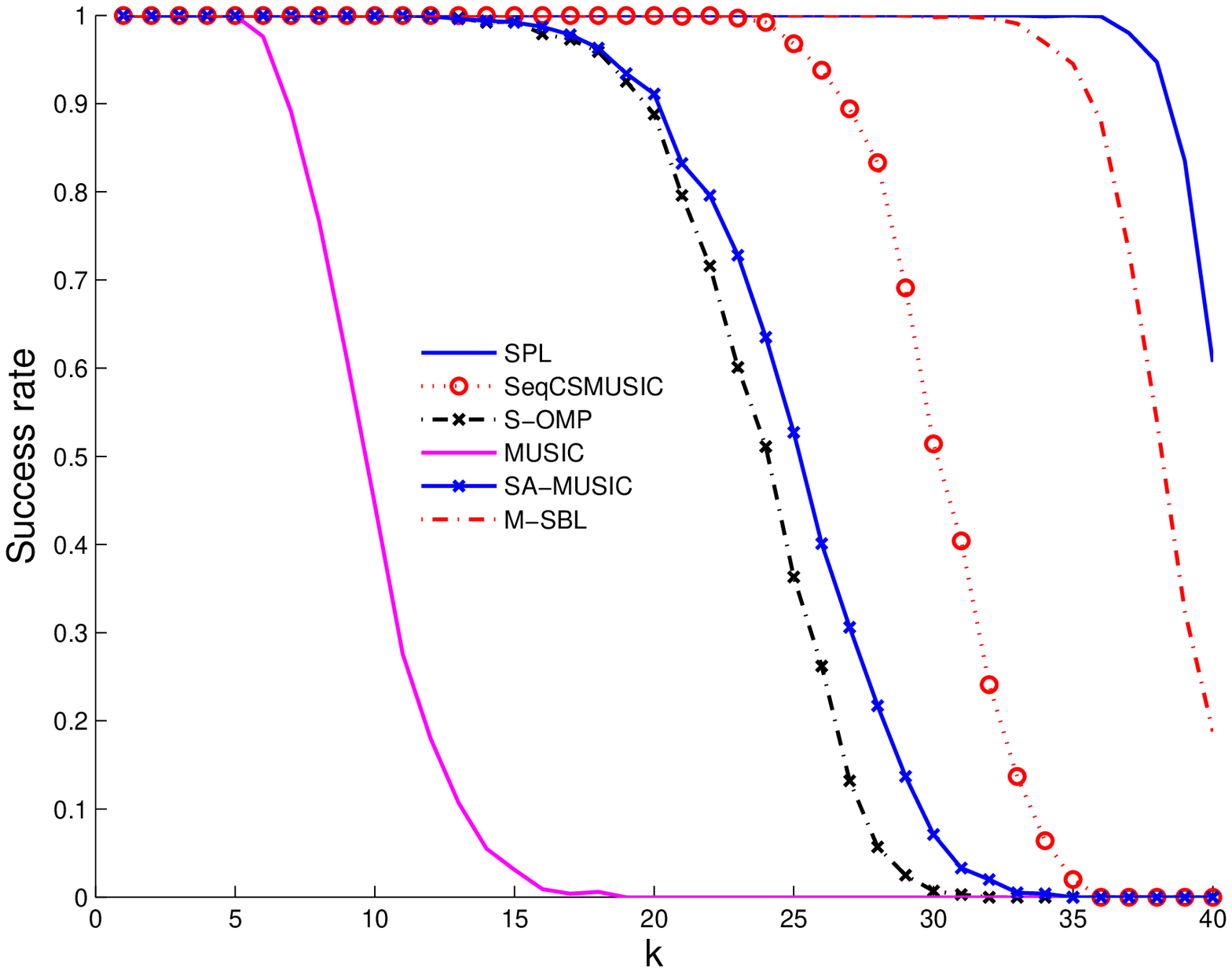}\includegraphics[width=7cm]{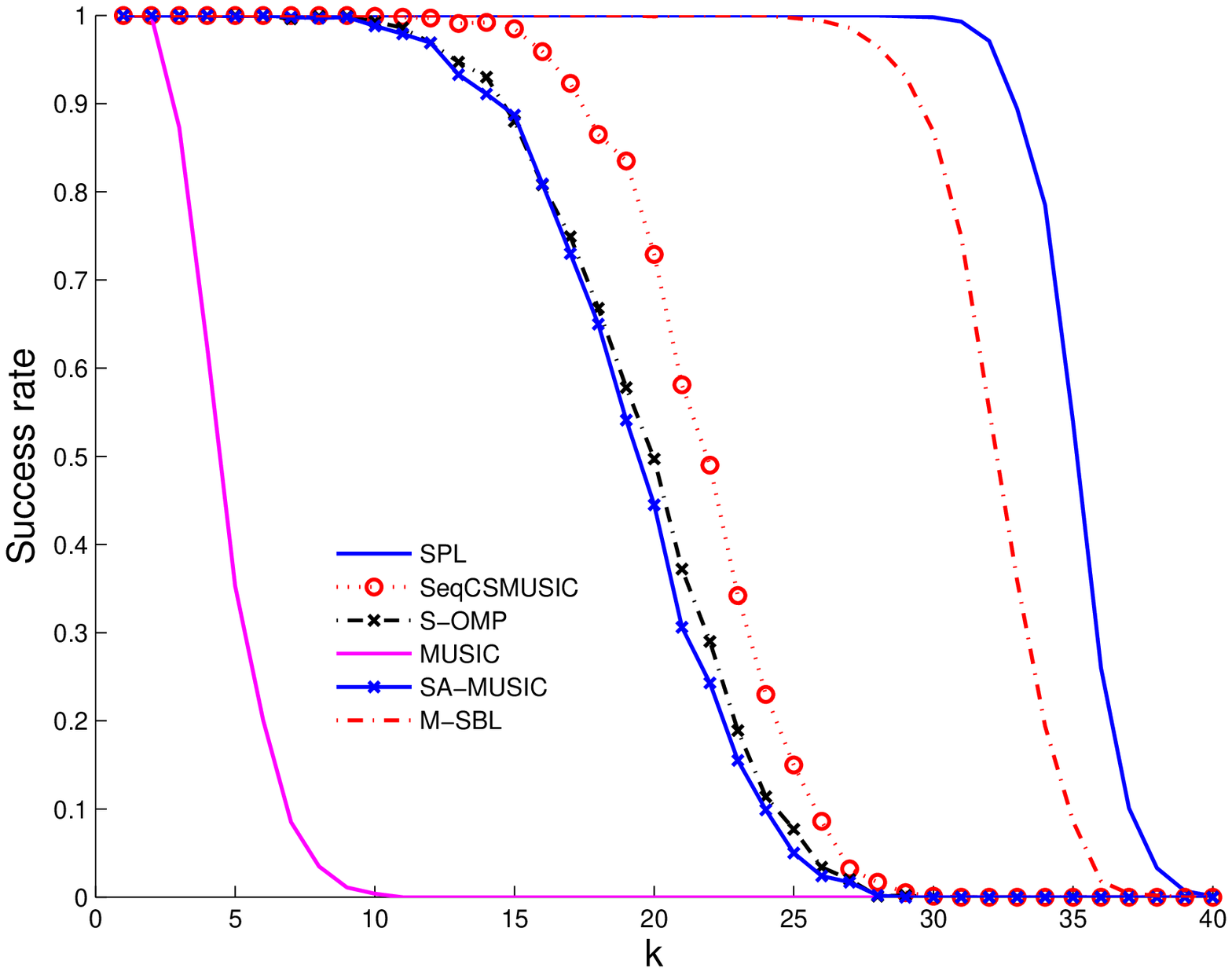}}\\
\centering{\mbox{(c)}\hspace{7cm}\mbox{(d)}}\vspace*{0.5cm}
\caption{Performance of various joint sparse recovery algorithms at $n=128$ when the sensing matrix is from Fourier matrix and SNR=30dB. The simulation parameters are
(a) $r=8, N=16, \tau=1, k=10$, (b) $r=8, N=256, \tau=1, k=10$, 
(c) $r=5, m=40, N=256, \tau=1$, (d)$r=15, m=40, N=256, \tau=0.5$, respectively.}
\label{fig:fourier}
\end{figure}

%
%
%
%
%

\section{CONCLUSION}
\label{sec:conclusion}

In this paper, we derived a new MMV algorithm called subspace penalized sparse learning (SPL) to address a joint sparse recovery problem, in which the unknown signals  share a common non-zero support.  The SPL algorithm was inspired by the observation that the $\log\det(\cdot)$  term in M-SBL is a rank proxy for a partial sensing matrix, and similar rank criteria exist in subspace-based greedy MMV algorithms like CS-MUSIC and SA-MUSIC.  Furthermore, we  proved that instead of  ${\rm rank}(A\Gamma^{1/2})$,  minimizing  ${\rm rank}(Q^*A\Gamma^{1/2})$ is a more direct way of imposing joint sparsity since its global minimizer can provide the true joint support. 
To impose such a subspace constraint as a penalty,  the SPL algorithm employs the Schatten-$p$ quasi norm rank penalty and was implemented as an alternating minimisation method.
Theoretical analysis showed that as $p\rightarrow 0$, 
the global minimizer  of the SPL is equivalent to the global minimiser of the $l_0$ MMV solution.
 We further demonstrated that compared to M-SBL,  our SPL is more robust to recovering badly conditioned $X_*$. 
 With numerical simulations, we demonstrated that SPL consistently outperforms all existing state-of-the art algorithms including M-SBL.

\section*{Acknowledgements}

This work was supported by Korea Science and Engineering Foundation under Grant NRF-2014R1A2A1A11052491.

\bibliographystyle{IEEEtran}
\bibliography{totalbiblio_bispl}
\end{document}